\documentclass{article}

\usepackage[frenchb,english]{babel}

\usepackage{amsmath}
\usepackage{amssymb}
\usepackage{amsthm}
\usepackage{amsfonts} 
\usepackage{stmaryrd}
\usepackage{graphicx}
\usepackage{comment}
\usepackage{dsfont}
\usepackage{xcolor}
\usepackage{mathtools}
\usepackage{bigints}
\usepackage{yhmath}
\usepackage{enumitem} 
\usepackage{subcaption} 
\usepackage{appendix} 
\usepackage{geometry} 
\usepackage{authblk} 
\usepackage{empheq} 
\usepackage{tikz} 

\newtheorem{propo}{Proposition}
\newtheorem{theor}{Theorem}

\newtheorem{defin}{Definition}
\newtheorem{remar}{Remark}

\newcommand{\vertii}[1]{{\left\vert\kern-0.3ex\left\vert #1 
    \right\vert\kern-0.3ex\right\vert}}

\newcommand{\vertiii}[1]{{\left\vert\kern-0.3ex\left\vert\kern-0.3ex\left\vert #1 
    \right\vert\kern-0.3ex\right\vert\kern-0.3ex\right\vert}}
    
\newcommand*\dd{\mathop{}\!\mathrm{d}}

\newcommand*\circled[1]{\tikz[baseline=(char.base)]{\node[shape=circle,draw,inner sep=2pt] (char) {#1};}}

\title{Properties of some dynamical systems for three collapsing inelastic particles}
\author{Th\'eophile Dolmaire}
\author{Juan J. L. Vel\'azquez}
\affil{Institute for Applied Mathematics, University of Bonn, Endenicher Allee 60, D-53115 Bonn, Germany}

\begin{document}

\maketitle

\begin{abstract}
\noindent
In this article we continue the study of the collapse of three inelastic particles in dimension $d \geq 2$, complementing the results we obtained in the companion paper \cite{DoVeAr1}. We focus on the particular case of the nearly-linear inelastic collapse, when the order of collisions becomes eventually the infinite repetition of the period \circled{0}-\circled{1}, \circled{0}-\circled{2}, under the assumption that the relative velocities of the particles (with respect to the central particle \circled{0}) do not vanish at the time of collapse. Taking as starting point the full dynamical system that describes two consecutive collisions of the nearly-linear collapse, we derive formally a two-dimensional dynamical system, called the two-collision mapping. This mapping governs the evolution of the variables of the full dynamical system. We show in particular that in the so-called Zhou-Kadanoff regime, the orbits of the two-collision mapping can be described in full detail. We study rigorously the two-collision mapping, proving that the Zhou-Kadanoff regime is stable and locally attracting in a certain region of the phase space of the two-collision mapping. We describe all the fixed points of the two-collision mapping in the case when the norms of the relative velocities tend to the same positive limit. We establish conjectures to characterize the orbits that verify the Zhou-Kadanoff regime, motivated by numerical simulations, and we prove these conjectures for a simplified version of the two-collision mapping.
\end{abstract}

\textbf{Keywords.} Inelastic Collapse; Inelastic Hard Spheres; Particle Systems.

\tableofcontents

\section{Introduction}
\numberwithin{equation}{section}

In this article we continue the study of a system of three inelastic hard spheres that we began in the companion paper \cite{DoVeAr1}. We consider the system in dimension $d \geq 2$. We assume that the particles, that we label \circled{$i$} with $i \in \{0,1,2\}$, are spheres of diameter $1$, moving freely when at a distance larger than $1$ from each other, and colliding if and only if the respective centers of two particles of the system are at distance $1$. In this case, the respective velocities $v_1$ and $v_2$ of the two colliding particles are instantaneously changed into $v_1' $ and $v_2'$ respectively, according to the law:
\begin{align*}
\left\{
\begin{array}{ccc}
v_1' + v_2' &=& v_1+v_2,\\
\left[ (v_1'-v_2')\cdot \omega_{1,2} \right] &=& - r \left[ (v_1-v_2)\cdot \omega_{1,2}\right],
\end{array}
\right.
\end{align*}
where $\omega_{1,2}$ is the unitary vector pointing from the center of the first particle towards the center of the second particle, and $r \in\, ]0,1[$ is a \emph{fixed} number, called the \emph{restitution coefficient}. From the above collision law, it is clear that the momentum of the system is conserved, but that a certain amount of kinetic energy is dissipated any time two particles collide.\\
\newline
Systems of inelastic hard spheres are extensively used as an idealized description, at the microscopic scale, of granular media (\cite{BrPo004}, \cite{CHMR021}). Such particle systems exhibit interesting behaviours (\cite{BeJN996}), such as clustering (\cite{GoZa993}, \cite{PoSc005}). Another typical phenomenon, known in the literature as \emph{inelastic collapse}, consists in the concentration of some particles of the system in small regions, that collide with each other more and more often, so that eventually an infinite number of collisions take place in finite time. The inelastic collapse was first observed in a one-dimensional model (\cite{BeMa990}), and later in higher dimension (\cite{McYo993}). This phenomenon causes challenging dynamical system problems. The properties of a system that experiences an inelastic collapse are still mainly unknown, even when such a system is constituted with a small number of particles. In addition, since such a phenomenon was observed numerically, it was expected to be stable, at least for some particular initial configurations, which we proved in \cite{DoVeAr1}. As a consequence, the well-posedness of the dynamics of such a particle system is a difficult problem, which remains mainly open. In comparison, such well-posedness is a well-known result in the case of elastic hard spheres, established by Alexander \cite{Alex975} (see also Proposition 4.1.1 in Section ``The $N$-particle flow'' of \cite{GSRT013} for a modern presentation). Let us also mention that in \cite{DoVeNot}, we obtained an Alexander's theorem for a different model of inelastic particles.\\
The literature concerning the inelastic collapse is rich, in particular concerning the one-dimensional case (\cite{ShKa989}, \cite{BeMa990}, \cite{McYo991}, \cite{CoGM995}, \cite{GrMu996}, \cite{GSBM998}, \cite{CDKK999}, \cite{BeCa999} and more recently \cite{ChKZ022} and \cite{HuRo023}). Concerning higher dimensions, the phenomenon is studied in \cite{ZhKa996} and \cite{ScZh996}. In the companion paper \cite{DoVeAr1} of the present article, we review extensively the literature about the inelastic collapse.\\
We obtained in the companion paper \cite{DoVeAr1} general results concerning collapsing systems of three inelastic particles, in dimension $d \geq 2$. We showed that only two orders of collisions are possible for a system of three collapsing particles. Namely, up to relabel the particles, the order of collisions becomes eventually the infinite repetition of the period \circled{0}-\circled{1}, \circled{0}-\circled{2}, or the infinite repetition of the period \circled{0}-\circled{1}, \circled{0}-\circled{2}, \circled{1}-\circled{2}. We named the former case the \emph{nearly-linear collapse}, and the latter the \emph{triangular collapse}, motivated by the final geometry of the particles at the time of collapse. The study of the triangular collapse in \cite{DoVeAr1} strongly suggests that the triangular collapse is unstable, although we do not have proved such a fact in \cite{DoVeAr1}. We proved that at the time when the collapse takes place, the three particles are in contact, and that the normal components of the relative velocities vanish. We wrote the full dynamical system governing the evolution of the particles during the collisions, and in particular we obtained asymptotics concerning the vanishing variables. Nevertheless, such a full dynamical system is extremely complicated, due to its high dimensionality: when we consider three particles evolving in $\mathbb{R}^2$, the system is $7$-dimensional, and $11$-dimensional for particles in $\mathbb{R}^3$. Therefore, one has to seek for simplifications, and to identify asymptotic regimes where some terms become negligible. Such an approach was already adopted by Zhou and Kadanoff in \cite{ZhKa996}, where two necessary conditions where derived for the existence and stability (with respect to perturbations of the initial data) of the inelastic collapse.\\
In particular, Zhou and Kadanoff identified a regime, that they named the ``flat surface approximation'', in which the full dynamical system can be reduced to a one-dimensional iteration, and the collapsing trajectories are identified to self-similar solutions of this iteration. However, there is no proof that the system of particles will necessarily enter such a regime, although this behaviour was observed numerically.\\
\newline
In general, one can show, using the asymptotic results in \cite{DoVeAr1}, that the full dynamical system can be formally reduced to a simplified two-dimensional system, as it was also observed in \cite{ZhKa996}. Understanding this two-dimensional system is the main objective of the present article. However, such a two-dimensional system might be different regarding the evolution of the system of particles, depending on the two following aspects: the order of collisions, that we discussed already, and the norm of the limit of the relative velocities. Concerning the order of collisions, it is important to know it a priori, in order to always iterate the same mapping, that describes the repetition of a given period of collisions. Concerning the norm of the relative velocities $v_1-v_0$ and $v_2-v_0$ (assuming that the particle \circled{0} is involved in all the collisions), each of them can be converging to zero, or not, at the limiting time of collapse. We expect that the generic case, that is the only case associated to a set of initial data of positive measure, corresponds to the case when none of the relative velocities vanish at the time of collapse. Such a statement can be motivated by the result of stability we obtained in \cite{DoVeAr1}, where we constructed a collapsing initial configuration, stable under perturbation of the initial datum, and such that the relative velocities do not vanish asymptotically. If on the contrary one assumes that one (or two) of the two relative velocities vanishes at the time of collapse, one would obtain a different dynamical system. Note that the case when one or two relative velocities vanish at the time of collapse constitutes a pathological configuration for the system of particles, in the sense that the dynamics of the system of particles becomes ill-posed after the time of collapse.\\
In the present article the simplified two-dimensional dynamical system is obtained with the following a priori assumptions on the system of particles: 
\begin{itemize}
\item first, we assume that the system of particles experiences indeed an inelastic collapse, so that we can consider infinitely many iterations of the mapping that represents a single collision,
\item second, we will assume that the sequence of collisions is the infinite repetition of the pair \circled{0}-\circled{1}, \circled{0}-\circled{2}, which corresponds to the \emph{nearly-linear collapse} following the nomenclature of \cite{DoVeAr1}, and that corresponds also to the situation studied in \cite{ZhKa996} (and which is, most likely, the only order of collisions associated to a set of initial data of positive measure),
\item third, we will to assume that none of the relative velocities $v_1-v_0$ and $v_2-v_0$ of the particles do not vanish in the limit, at the time of collapse.
\end{itemize}
As we shall see, this two-dimensional system turns out to be a complicated mathematical object, even though simpler than the full dynamical system of the three particles system. Let us denote by $\eta_1$ and $\eta_2$ the respective normal components of the two relative velocities $v_1-v_0$ and $v_2-v_0$, and by $\tau$ the time between two consecutive collisions. Essentially, the two-dimensional system describes the evolution of the variables $\eta_1$, $\eta_2$ and $\tau$, which in turn govern the evolution of all the other variables. However, the so-called \emph{Zhou-Kadanoff regime}, that we described in \cite{DoVeAr1}, allows a complete study of the two-dimensional system. In this regime, which corresponds to the case when $\tau \ll \eta_1,\eta_2$ asymptotically, the two-collision mapping reduces to a one-dimensional mapping. This regime corresponds also to the flat surface approximation of Zhou and Kadanoff. As a consequence, it is a central question to know under which conditions the Zhou-Kadanoff regime holds, since it would ultimately provide important information concerning the asymptotic behaviour of a system of three inelastic particles in the neighbourhood of the collapse. Therefore, we attached ourselves in this article to determine when the Zhou-Kadanoff regime holds.\\
\newline
The main results of the present article are the following. We first derive formally the two-dimensional reduction of the full dynamical system which describes two consecutive collisions in the case of the nearly-linear collapse, assuming that the norms of the two relative velocities are not vanishing at the time of collapse. The pair of variables $\left(\varphi_1,\varphi_2\right)$ of this two-dimensional reduction are described in terms of $\eta_1$, $\eta_2$ and $\tau$, and we identify the Zhou-Kadanoff regime as the regime in which $\varphi_2$ converges to zero.\\
Then, we describe the orbits of this two-dimensional system using a combination of analytical arguments and numerical simulations. In particular, we prove rigorously that the Zhou-Kadanoff regime is always attracting in a non-trivial region of the phase space of the two-dimensional reduction. In the symmetric case, that is, when the norms of the relative velocities converge to the same positive limit at the time of collapse, we describe all the fixed points of the two-dimensional dynamical system, and we study their stability. Motivated by numerical simulations, we conjecture a characterization of the orbits that satisfy eventually the Zhou-Kadanoff regime. Finally, we derive a simplified version of the two-dimensional reduction, corresponding to the limit when the energy of the system is small at the time of collapse, and we prove the conjecture for this simplified system.\\
\newline
Let us describe the results in more detail. We start with recalling the general results obtained in the companion paper \cite{DoVeAr1} concerning the collapsing systems of three inelastic particles. In particular, in the regime of the collapse, we recall in which sense the variables used to describe such a collapse are all estimated by the normal components $\eta_1$ and $\eta_2$ of the relative velocities and the time difference $\tau$ between two consecutive collisions. In this article we apply these results to the case of the nearly-linear collapse with non-vanishing relative velocities at the time of collapse. We formally obtain the evolution laws of the leading order terms of the full dynamical system describing the evolution of the three particles, which are given only in terms of the three variables $\eta_1$, $\eta_2$ and $\tau$. We show how such evolution laws can be reduced to a two-dimensional dynamical system, acting on the variables $\left(\varphi_1,\varphi_2\right)$, that we define in the article. We call this dynamical system the \emph{two-collision mapping}. We identify also formally the Zhou-Kadanoff regime $\varphi_2 \rightarrow 0$ for which the system simplifies into a one-dimensional iteration that can be fully studied. These results constitute the formal study of the full dynamical system describing the evolution of the three inelastic particles.\\
From this point, we perform a mathematically rigorous study of the two-collision mapping, investigating the consequences of the Zhou-Kadanoff regime, and when such a regime holds. In particular, we prove that the region $\varphi_2 = 0$ is invariant, and we give a complete description of the dynamics of the two-collision mapping in this region. We prove that two fixed points exist in this region, one of them being stable for orbits starting from $\varphi_2=0$. Then, we prove that this region is actually attracting in the plane $\left(\varphi_1,\varphi_2\right)$ close to the stable fixed point in $\varphi_2=0$, which proves that the Zhou-Kadanoff regime is indeed stable for some initial configurations of the two-collision mapping.\\
In order to go beyond the local analysis of the stability of the region $\varphi_2=0$ and to characterize the basin of attraction of the Zhou-Kadanoff regime, we report numerical observations we made on the two-collision mapping. These observations enabled us to state conjectures, and we believe in particular that there exists a separatrix in the plane $\left(\varphi_1,\varphi_2\right)$, which divides the orbits between only two groups (except for the orbits that are contained in the separatrix itself): the orbits that are globally well-defined and for which the Zhou-Kadanoff regime eventually holds, and the orbits that are not globally well-defined. In summary, we conjecture a complete description of the phase space of the two-collision mapping, claiming that the Zhou-Kadanoff is the only regime that is associated to a set of initial data of positive measure for the two-collision mapping.\\
We then turn to a symmetric version of the collision mapping, corresponding to the case when the norms of the relative velocities tend to the same positive limit at the time of collapse. In this case, the two-collision mapping simplifies, in such a way that it corresponds to two iterations of the same dynamical system, the so-called \emph{one-collision mapping}. We provide therefore a study of the one-collision mapping. We describe with full mathematical rigor all the fixed points of the one-collision mapping. We recover in particular the unstable fixed point already observed by Zhou and Kadanoff in \cite{ZhKa996}, and we believe that this point lies on the separatrix that characterizes the Zhou-Kadanoff regime. This fixed point lies outside the region $\varphi_2=0$, and in particular for such a fixed point the Zhou-Kadanoff regime does not hold.\\
Finally, we derive the equations of the so-called \emph{low energy limit} of the two-collision mapping, which corresponds to a situation of a collapse when the remaining kinetic energy of the system is small (but non-zero) at the time of collapse. For this system, we prove the conjecture we stated for the two-collision mapping: we show that there exists a separatrix, below which the orbits are globally well-defined and eventually verify the Zhou-Kadanoff regime, and above which the orbits are not globally well-defined. In particular, we can describe completely the orbits of the two-collision mapping in the low energy limit.\\
Some of the results that we describe in this article have analogies with what Zhou and Kadanoff obtained in \cite{ZhKa996}. Nevertheless, one of the main novelties of the present article, and that is not in \cite{ZhKa996}, consists in the complete description of the orbits of the two-collision mapping, including the existence of the separatrix. For the general two-collision mapping, we conjecture such a description, and we prove it in this article only in the formal limit of the low energy.\\
\newline
The plan of the paper is the following. In the second section, we provide the expression of the full dynamical system describing the three inelastic particles obtained in the companion paper \cite{DoVeAr1}, and we recall the main asymptotic results concerning such a system, from the same article. In the third section, we derive and study the expression of the two-collision mapping, which is the simplified two-dimensional reduction of the full dynamical system. We establish the main properties of the two-collision mapping, we prove that the Zhou-Kadanoff regime is asymptotically stable, at least in a non trivial region of the two-dimensional phase space, and we conclude this section with the conjectures on the behaviour of the orbits of the two-collision mapping. In the fourth section, we describe all the fixed points of the one-collision mapping. Finally, in the fifth section, we present the formal derivation of the equations of the two-collision mapping in the \emph{low energy limit}, and we study completely this simplified dynamical system.

\section{Obtaining the complete dynamical system}
\label{SECTION__2EcrirSystmDynam}

\subsection{The model}
\label{SSECTI02.1_Le_Modele}

We consider a system of three inelastic particles in the Euclidean space $\mathbb{R}^d$ ($d \geq 2$). Such particles will be denoted by \circled{$i$} ($i \in \{0,1,2\}$). Let us denote $x_i \in \mathbb{R}^d$ and $v_i \in \mathbb{R}^d$ the respective positions and velocities of the three particles ($i \in \{0,1,2\}$). We assume that the particles are identical spheres (that is, of the same diameter, equal to $1$) that cannot overlap. We will assume that the particles evolve according to the \emph{inelastic hard sphere flow}, defined as follows.\\
When no pair of particles are in contact, that is, when $\vert x_i - x_j \vert > 1$ for all $0 \leq i<j \leq 2$, the particles move in straight line, with constant velocities.\\
When two particles collide, that is, when there exists a time $t_1 \geq 0$ such that $\vert x_i(t_1) - x_j(t_1) \vert = 1$ for a certain pair of indices $i<j$, the velocities of the two particles \circled{$i$} and \circled{$j$} are immediately changed from $v_i = v_i(0)$ and $v_j = v_j(0)$ into
\begin{equation}
\label{EQUATSS2.1VitesPost-Colli}
\left\{
\begin{array}{rl}
v_i' &= v_i - \frac{(1+r)}{2}(v_i-v_j)\cdot \omega_{i,j} \omega_{i,j},\\
v_j' &= v_j + \frac{(1+r)}{2}(v_i-v_j)\cdot \omega_{i,j} \omega_{i,j},
\end{array}
\right.
\end{equation}
where $r \in [0,1]$ is the \emph{restitution coefficient}, and $\omega_{i,j}$ is the normalized vector $\left(x_i(t_1)-x_j(t_1)\right)/\vert x_i(t_1)-x_j(t_1) \vert$ joining the centers of the two colliding particles \circled{$i$} and \circled{$j$}, at the time of the collision $t_1$.\\
\newline
The trajectories of systems of inelastic particles are defined provided that only binary collisions take place between free evolution of the particles. As for the classical elastic hard sphere system, triple collisions (that is, when three or more particles collide all together exactly at the same time) yields an ill-posed problem. An additional difficulty, specific to the case of inelastic particles, and that appears in particular for models with a fixed restitution coefficient $r$, is the possibility of infinitely many collisions to take place in a finite time. Such a phenomenon, called inelastic collapse, prevents also to define the dynamics of the particles further. The question of the global well-posedness of such a flow, for a general initial configuration, is still an open problem, and is one of the motivations of the present paper.

\subsection{Parametrizing the collisions}
\label{SSECTI02.2ParamColli}

Let us  now parametrize the system of three particles, and define the discrete dynamical system that acts on such parametrization, describing the evolution of the system from a collision to the following one. In particular, we will write the equations describing a first collision occuring at a positive time $\tau$, between the particles \circled{0} and \circled{2}, assuming that \circled{0} and \circled{1} just collided at the initial time $t=0$. Proceeding that way, any collision can be described with the same equations, up to relabel the particles.\\
At time $t=0$, we have:
\begin{align}
\left\vert x_1-x_0 \right\vert = 1,
\end{align}
corresponding to the fact that \circled{0} and \circled{1} are in contact. Let us then denote by $\omega_1$ the difference $x_1 - x_0 \in \mathbb{S}^{d-1}$. We have also
\begin{align}
\label{EQUATSS2.2LoiEvW1Om1Init2}
W_1 \cdot \omega_1 > 0,
\end{align}
introducing then the relative velocity $W_1 = v_1 - v_0$ of the particle \circled{1} with respect to \circled{0}, describing the fact that the two particles \circled{0} and \circled{1} are in a post-collisional configuration. Let us also assume that the particle \circled{2} is at a positive distance from the two others at the initial time. We have:
\begin{align}
d > 0.
\end{align}
where we denote by $\omega_2$ the vector $\displaystyle{\frac{x_2-x_0}{\vert x_2 - x_0 \vert}} \in \mathbb{S}^{d-1}$, and $1+d = \vert x_2 - x_1 \vert$, $x_2 - x_0 = (1+d) \omega_2$. Denoting the normal components of the relative velocities $W_i\cdot\omega_i$ by $\eta_i$, we assume finally:
\begin{align}
\label{EQUATSS2.2LoiEvW2Om2Init1}
\eta_2 = W_2 \cdot \omega_2 < 0,
\end{align}
where we denoted by $W_2$ the relative velocity $v_2 - v_0$, and:
\begin{align}
\label{EQUATSS2.2LoiEvParZK_zeta}
\frac{d(2+d) \vert W_2 \vert^2}{(1+d)^2 \eta_2^2} \geq 0.
\end{align}
The two conditions \eqref{EQUATSS2.2LoiEvW2Om2Init1} and \eqref{EQUATSS2.2LoiEvParZK_zeta} together constitute necessary and sufficient conditions for a collision between the particles \circled{0} and \circled{2} to take place in the future. When these two conditions hold, the collision between \circled{0} and \circled{2} takes place at time $\tau$, where:
\begin{align}
\label{EQUATSS2.2TempsColliTauV1}
\tau &= - \frac{(1+d)\left( \omega_2 \cdot W_2 \right)}{\vert W_2 \vert^2} - \frac{\sqrt{ (1+d)^2 \left( \omega_2 \cdot W_2 \right)^2 - d(2+d) \vert W_2 \vert^2 }}{\vert W_2 \vert^2} \nonumber\\
&= \frac{(1+d)(-\eta_2)}{\vert W_2 \vert^2} \left[ 1 - \sqrt{1 - \frac{d(2+d) \vert W_2 \vert^2}{(1+d)^2 \eta_2^2}} \right].
\end{align}
In particular, the quantity \eqref{EQUATSS2.2LoiEvParZK_zeta}, involved in the expression \eqref{EQUATSS2.2TempsColliTauV1} of $\tau$ is of central importance, in particular to characterize the Zhou-Kadanoff regime.

\begin{defin}[Zhou-Kadanoff parameter]
Let $d \geq 2$ be a positive integer, and $W_2$ be a vector in $\mathbb{R}^d$. We denote by $\zeta$ the positive quantity:
\begin{align}
\label{EQUATSS2.2DefinParZK_zeta}
\zeta = \frac{d(2+d) \vert W_2 \vert^2}{(1+d)^2 \eta_2^2} \geq 0.
\end{align}
The number $\zeta$ will be called the \emph{Zhou-Kadanoff parameter} (in short, the \emph{ZK parameter}).
\end{defin}
\noindent
With the variables we introduced, we can parametrize the system of three inelastic particles as:
\begin{align}
\big( \hspace{-2mm} \underbrace{\omega_1,W_1}_{\substack{\text{position and}\\ \text{velocity of}\\ \text{the particle}\\ \circled{1}}},\underbrace{d,\omega_2,W_2}_{\substack{\text{position and}\\ \text{velocity of}\\ \text{the particle}\\ \circled{2}}} \hspace{-0.5mm} \big) \in \mathbb{S}^{d-1} \times \mathbb{R}^d \times \mathbb{R}_+^* \times \mathbb{S}^{d-1} \times \mathbb{R}^d.
\end{align}
We need in total $4d-1$ real variables in order to describe the initial configuration of the system. We can also consider how such variables evolve when a collision between the particles \circled{0} and \circled{2} takes place. The evolution law of these variables are summarized in the mapping we introduce in the following definition.

\begin{defin}[Complete one-collision mapping]
Let us consider three particles \circled{0}, \circled{1}, \circled{2} in $\mathbb{R}^d$, of respective positions $x_i$ and velocities $v_i$ ($x_i,v_i \in \mathbb{R}^d \ \forall i \in \{ 1, 2, 3 \}$) described by the configuration:
\begin{align}
\left(\omega_1,W_1,d,\omega_2,W_2\right) \in \mathbb{S}^{d-1} \times \mathbb{R}^d \times \mathbb{R}_+^* \times \mathbb{S}^{d-1} \times \mathbb{R}^d,
\end{align}
where:
\begin{itemize}
\item $\omega_1 = x_1-x_0$ ( \circled{0} and \circled{1} are initially in contact),
\item $W_1 = v_1-v_0$,
\item $1+d = \vert x_2 - x_0 \vert > 1$ ( \circled{0} and \circled{2} are initially separated),
\item $(1+d) \omega_2 = x_2-x_0$,
\item $W_2 = v_2-v_0$.
\end{itemize}
Let us assume in addition that \circled{0} and \circled{1} are in a post-collisional configuration, and that a collision of type \circled{0}-\circled{2} is the next collision that will take place:
\begin{align}
\eta_1 = W_1 \cdot \omega_1 > 0,\hspace{5mm} \eta_2 = W_2 \cdot \omega_2 < 0, \hspace{2mm} \text{and} \hspace{2mm} \zeta = \frac{d(2+d) \vert W_2 \vert^2}{(1+d)^2\eta_2^2} < 1.
\end{align}
We define then the \emph{complete one-collision mapping} as the function:
\begin{align*}
\mathfrak{C} :
\left\{
\begin{array}{clc}
\mathbb{S}^{d-1} \times \mathbb{R}^d \times \mathbb{R}_+^* \times \mathbb{S}^{d-1} \times \mathbb{R}^d &\rightarrow &\mathbb{S}^{d-1} \times \mathbb{R}^d \times \mathbb{R}_+^* \times \mathbb{S}^{d-1} \times \mathbb{R}^d, \\
\left(\omega_1,W_1,d,\omega_2,W_2\right) &\mapsto &\left(\omega_1',W_1',d',\omega_2',W_2'\right),
\end{array}
\right.
\end{align*}
where:
\begin{align}
\label{EQUATSS2.2IterationSystm1}
\left\{
\begin{array}{rl}
\omega_1' &= \displaystyle{\frac{\omega_1 + \tau W_1}{(1+d')}},\vspace{1mm} \\
\eta_1' &= \displaystyle{\frac{1}{(1+d')} \left( \eta_1 + \tau \vert W_1 \vert^2 \right) - \frac{(1+r)}{2} \left( \omega_1'\cdot\omega_2' \right) \left( (1+d)\eta_2 + \tau \vert W_2 \vert^2 \right)}, \vspace{3mm}\\
(W_1^\perp)' &=W_1 - W_1\cdot \omega_1' \omega_1' + \displaystyle{\frac{(1+r)}{2}} \left( W_2\cdot\omega_2' \right) \big[ (\omega_1'\cdot\omega_2') \omega_1' - \omega_2' \big] \vspace{1mm}\\
&\hspace{20mm}= W_1 - \displaystyle{\frac{(\eta_1 + \tau \vert W_1 \vert^2)}{(1+d')}} \omega_1' + \displaystyle{\frac{(1+r)}{2}} \left( (1+d)\eta_2 + \tau \vert W_2 \vert^2 \right) \big[ (\omega_1'\cdot\omega_2') \omega_1' - \omega_2' \big], \vspace{3mm}\\
d' &= \displaystyle{\sqrt{ 1 + 2 \eta_1 \tau + \vert W_1 \vert^2 \tau^2} - 1}, \vspace{3mm}\\
\omega_2' &= (1+d)\omega_2 + \tau W_2, \vspace{3mm}\\
\eta_2' &= -r(1+d) \eta_2 - r\tau \vert W_2 \vert^2, \vspace{3mm}\\
(W_2^\perp)' &= W_2 - (W_2\cdot\omega_2')\omega_2' \vspace{2mm}\\
&\hspace{20mm}= W_2 - \left( (1+d)\eta_2 + \tau \vert W_2 \vert^2 \right) \left( (1+d)\omega_2 + \tau W_2 \right),
\end{array}
\right.
\end{align}
\noindent
with $\omega_1'\cdot \omega_2'$ defined as:
\begin{align}
\label{EQUATSS2.2CosinAngleVers1}
\omega_1'\cdot\omega_2' &= \left( \frac{\omega_1+\tau W_1}{(1+d')} \right) \cdot \left( (1+d)\omega_2+\tau W_2 \right) \nonumber\\
&= \frac{(1+d)}{(1+d')} \left( \omega_1 \cdot \omega_2 \right) + \left[ \frac{1}{(1+d')} \omega_1 \cdot W_2 + \frac{(1+d)}{(1+d')}\omega_2 \cdot W_1 \right] \tau + \frac{W_1 \cdot W_2}{(1+d')} \tau^2,
\end{align}
and $\tau$ is defined as
\begin{align}
\label{EQUATSS2.2TempsColliTauV3}
\tau = \frac{(1+d)(-\eta_2)}{\vert W_2 \vert^2} \left[ 1 - \sqrt{1 - \zeta} \right], \text{  and  } 
\zeta = \frac{d(2+d) \vert W_2 \vert^2}{(1+d)^2 \eta_2^2} \geq 0.
\end{align}
\end{defin}
\noindent
The mapping $\mathfrak{C}$ introduced in the previous definition encodes completely the dynamics of the system of three inelastic particles, describing the evolution of the particles between a collision of type \circled{0}-\circled{1} and a collision of type \circled{0}-\circled{2}. All the distances and velocities are measured from the position $x_0$ and velocity $v_0$ of the particle \circled{0}, central for the pair of collisions \circled{0}-\circled{1}, \circled{0}-\circled{2}.

\begin{remar}
$\mathfrak{C}$ defines a mapping from $\mathbb{S}^{d-1} \times \mathbb{R}^d \times \mathbb{R}_+^* \times \mathbb{S}^{d-1} \times \mathbb{R}^d$ into itself: we have a $(4d-1)$-dimensional dynamical system. Note that the system is $7$-dimensional for $d=2$, and $11$-dimensional for $d=3$.
\end{remar}

\subsection{The main properties of the nearly-linear inelastic collapse}
\label{SSECTIO2.3ProprGenerColla}

As we indicated already in the introduction, the present paper is focused on studying systems of three inelastic particles undergoing a particular type of collapse. In the companion paper \cite{DoVeAr1}, we proved that only two collision orders are possible for a system of three inelastic particles: up to relabel the particles, the collision order becomes eventually, either the infinite repetition of the period \circled{0}-\circled{1}, \circled{0}-\circled{2}, \circled{1}-\circled{2} (we called such a configuration a \emph{triangular collapse}, and we expect that such a configuration is unstable under perturbations of the initial data in the phase space), or the infinite repetition the period \circled{0}-\circled{1}, \circled{0}-\circled{2}. We focus our attention here on the later period, already studied by Zhou and Kadanoff (\cite{ZhKa996}). We proved in the companion paper \cite{DoVeAr1} that such a configuration is stable under perturbations of the initial data in the phase space, at least for some particular initial data. We called such a collapsing configuration a \emph{nearly-linear inelastic collapse}.\\
In the case of a nearly-linear collapse where \circled{0} plays the role of the central particle (the only particle involved in all the collisions), the normal components of the relative velocities $W_1 = v_1-v_0$ and $W_2 = v_2-v_0$ vanish, but there is no reason to believe that the tangential components vanish as well. On the contrary, we believe that in the generic case (that is, the only case associated to a set of initial data with a positive measure in the phase space), the norms of the relative velocities $W_1$ and $W_2$ are non zero at the limiting time of the collapse. We will therefore study only such configurations. As in \cite{ZhKa996}, we study the systems of three collapsing inelastic particles, with a priori assumptions on the final state of such systems.\\
The purpose of this section is to recall general results concerning the nearly-linear collapse, established in full details in the companion paper \cite{DoVeAr1}. Let us start with the definition of the nearly-linear collapse.

\begin{defin}[Nearly-linear inelastic collapse]
\label{DEFINSS3.1Collapse_Inelas}
Let $r \in\ ]0,1[$ be a positive real number smaller than $1$, and let us consider a system of three particles \circled{0}, \circled{1} and \circled{2} evolving according to the $r$-inelastic hard sphere flow introduced in Section \ref{SSECTI02.1_Le_Modele}, on a time interval $[0,\tau^*[$, with $\tau^* > 0$.\\
We say that the system undergoes a \emph{nearly-linear inelastic collapse} at the time $\tau^*$, called the \emph{time of the collapse}, or the \emph{collapsing time}, if there exists an increasing sequence of positive times $(t_n)_{n\in\mathbb{N}}$, with $t_0 = 0$, such that $t_n \xrightarrow[n \rightarrow + \infty]{}\tau^*$, where the times $t_n$ correspond exactly to the times of collisions between the particles, such that $\sup_{n}t_n$ is the only accumulation point of the sequence $\left(t_n\right)_{n\in\mathbb{N}}$, and in addition such that the sequence of collisions becomes eventually the infinite repetition of the pairs of collisions \circled{0}-\circled{1}, \circled{0}-\circled{2}.\\
In particular, for any $t \in \ ]t_n,t_{n+1}[$, the system of particles evolves according to the free flow.\\
The particle \circled{0}, involved in infinitely collisions with the two other particles \circled{1} and \circled{2}, will be called the \emph{central particle}. The particles \circled{1} and \circled{0}, involved in infinitely many collisions only with the particle \circled{0}, will be called, alternatively, the \emph{colliding}, and the \emph{spectator} particles.
\end{defin}

\subsubsection{Converging quantities for the nearly-linear inelastic collapse}

Let us now present the elementary properties of a system of three inelastic particles that undergoes a nearly-linear inelastic collapse. We summarize the results of \cite{DoVeAr1} in the following proposition.

\begin{propo}[Vanishing and converging variables in the case of a nearly-linear collapse]
\label{PROPOSS3.2ConveInterTemps}
Let $r \in\ ]0,1[$ be a positive number smaller than $1$, and let us consider a system of three particles \circled{0}, \circled{1} and \circled{2} evolving according to the $r$-inelastic hard sphere flow, on a time interval $[0,\tau^*[$, and undergoing a nearly-linear inelastic collapse at time $\tau^* > 0$.\\
Then, the following properties hold true.
\begin{itemize}
\item The sequence $\left( \tau_n \right)_n$ of the intervals $\tau_n = t_n - t_{n-1}$ between two consecutive collisions is summable, that is,  $\left( \tau_n \right)_n \in\ \ell^1$. In particular, we have:
\begin{align}
\tau_n \xrightarrow[n \rightarrow +\infty]{} 0.
\end{align}
\item The particles \circled{0} and \circled{1}, resp. \circled{0} and \circled{2}, are in contact at the collapsing time $\tau^*$, that is, the relative positions $x_1(t_n)-x_0(t_n)$ and $x_2(t_n)-x_0(t_n)$ converge respectively towards $\overline{\omega}_1$ and $\overline{\omega}_2$ as $n \rightarrow +\infty$, and we have:
\begin{align}
\vert \overline{\omega}_1 \vert = 1 \hspace{5mm} \text{and} \hspace{5mm} \vert \overline{\omega}_2 \vert = 1.
\end{align}
\item We have:
\begin{align}
\eta_{1,n},\eta_{2,n} \xrightarrow[n \rightarrow +\infty]{} 0,
\end{align}
where $\eta_{i,n}$ denotes the normal component of the relative velocity between the particles \circled{0} and \circled{i}, that is $\eta_{i,n} = W_i\cdot \omega_i = (v_i-v_0)\cdot(x_i-x_0)$.\\
In addition, the normal components $\eta_{1,n}$ and $\eta_{2,n}$ converge exponentially fast to zero, at a rate at least equal to
\begin{align*}
\max\left(\frac{(1+r)}{2} \big\vert \cos\overline{\theta} \big\vert,r\right),
\end{align*}
where $\overline{\theta}$ denotes the angle between the angular parameters $\omega_1 = x_1-x_0$ and $\omega_2 = x_2-x_0$ at the limiting time $\tau^*$ of the collapse. As a consequence, the series of the normal components $\sum_{n \geq 0} \vert \eta_{1,n} \vert$ and $\sum_{n \geq 0}\vert \eta_{2,n} \vert$ are both converging.
\item The sequences of the tangential components $\left(W_1^\perp\right)_{n\in\mathbb{N}}$ and $\left(W_2^\perp\right)_{n\in\mathbb{N}}$ of the relative velocities $W_1 = v_1-v_0$ and $W_2 = v_2-v_0$ are converging as $n\rightarrow+\infty$.
\end{itemize}
\end{propo}
\noindent
When the system experiences a nearly-linear inelastic collapse, the particles arrange in a nearly-linear chain and are in contact at the limiting time $\tau^*$ of collapse, hence the name of the configuration. In the middle of this chain lies the central particle \circled{0}. The two external particles \circled{1} and \circled{2} form an angle $\overline{\theta} = \left(\overline{\omega_1},\overline{\omega_2}\right)$. We will see that such an angle cannot be smaller than $\pi/2$.\\
According to Proposition \ref{PROPOSS3.2ConveInterTemps}, the relative velocities are converging at the limiting time of the collapse. Let us then denote by $a$ and $b$ the squares of the respective norms of the relative velocities $W_1$ and $W_2$ at the collapsing time $\tau^*$, that is:
\begin{align}
\left\vert W_{1,n} \right\vert \xrightarrow[n \rightarrow +\infty]{} a \hspace{3mm} \text{and} \hspace{3mm} \left\vert W_{2,n} \right\vert \xrightarrow[n \rightarrow +\infty]{} b.
\end{align}
We will assume throughout this paper that $a,b > 0$.\\
\newline
We now have a clear picture of the configuration of the system at the limiting time of collapse $\tau^*$. Let us recall that the proofs of all the results presented above are performed in \cite{DoVeAr1}.

\subsubsection{The results of Zhou and Kadanoff}
\label{SSSCT2.3.2ResulZhou&Kadan}

We recall here the only results, to the best of our knowledge, existing in the literature (except the companion paper \cite{DoVeAr1}) concerning the collapse of three inelastic particles in dimension $d \geq 2$. The results are due to Zhou and Kadanoff (\cite{ZhKa996}), and they describe, for a system experiencing a nearly-linear inelastic collapse, conditions on the angle $\overline{\theta}$ (formed by the particles at the limiting time of collapse) in terms of the restitution coefficient $r$, in order to have existence and stability of the collapse. The results can be summarized in the two following points.
\begin{itemize}
\item The existence of such a nearly-linear inelastic collapse is possible only if:
\begin{align}
\label{EQUATSS3.4ZKCl1CondiExist}
-\cos \theta \geq \frac{4\sqrt{r}}{1+r} \cdotp
\end{align}
\item In addition, such a nearly-linear inelastic collapse is stable (with respect to perturbations of the initial data leading to such a collapse) only if:
\begin{align}
\label{EQUATSS3.4ZKCl2CondiStabi}
-\cos \theta > \frac{2r^{1/3}(1+r^{1/3})}{1+r} \cdotp
\end{align}
\end{itemize}
In particular, no nearly-linear inelastic collapse can exist if the restitution coefficient is larger than $r_\text{exist.}$, and no such collapse can be stable if the restitution coefficient is larger than $r_\text{stabi.}$, where:
\begin{align}
r_\text{exist.} = 7-4\sqrt{3} \simeq 0.07179677 \hspace{5mm} \text{and} \hspace{5mm} r_\text{stabi.} = 9 - 4\sqrt{5} \simeq 0.05572809.
\end{align}

\subsection{Asymptotic behaviour of the variables describing the full dynamical system}
\label{SSECTIO2.4CompoAsympVaria}

In this section, we will investigate the asymptotic behaviour of the variables of the dynamical system \eqref{EQUATSS2.2IterationSystm1} to describe a collision, assuming that the system of three inelastic particles experiences a nearly-linear inelastic collapse. As we mentioned above already, we will also assume that the relative velocities $W_1$ and $W_2$ will not vanish at the time of the collapse. Let us start with some notations.

\subsubsection{Notations}

In the rest of this section, concerning a variable $y$, we will denote by
\begin{align*}
y = \mathcal{O}\left(x_1,\dots,x_n\right)
\end{align*}
when this variable $y$ is bounded by the $n$ other variables $x_1,\dots,x_n$ as $y,x_1,\dots,x_n \rightarrow 0$, that is, if there exist $n$ positive constants $C_1,\dots,C_n > 0$ such that $
\vert y \vert \leq C_1 \vert x_1 \vert + \dots C_n \vert x_n \vert$ for $y,x_1,\dots,x_n$ in a neighbourhood of $0$. If the two variables $y$ and $z$ are bounded by each other as $y,z \rightarrow 0$, that is, if there exist two positive constants $C_1,C_2 > 0$ such that $\vert y \vert \leq C_1 \vert z \vert$ and $\vert z \vert \leq C_2 \vert y \vert$ for $y,z$ in a neighbourhood of $0$, such a situation will be denoted by
\begin{align*}
y \lessgtr z.
\end{align*}
In the same way, we will denote by
\begin{align*}
y = o\left(x_1,\dots,x_n\right)
\end{align*}
when the variable $y$ is negligible with respect to the $n$ variables $x_1,\dots,x_n$ as $x_1,\dots,x_n \rightarrow 0$, that is, if there exist $n$ functions $\varphi_1,\dots,\varphi_n:\mathbb{R}_+\rightarrow \mathbb{R}_+$ vanishing at $0$ and such that $\vert y \vert \leq \varphi_1\left(\vert x_1 \vert\right) + \dots + \varphi_n\left(\vert x_n \vert\right)$, or equivalently, if there exists a function $\varphi:\mathbb{R}_+\rightarrow \mathbb{R}_+$ vanishing at $0$ such that $\vert y \vert \leq \varphi\left( \vert x_1 \vert + \dots \vert x_n \vert \right)$ for $y,x_1,\dots,x_n$ in a neighbourhood of $0$.

\subsubsection{Asymptotic comparison of the variables, for non vanishing tangential velocities}

As a matter of fact, all the vanishing quantities involved in the description \eqref{EQUATSS2.2IterationSystm1} of a collapsing system of particles can all be estimated with three variables only: the normal components $\eta_1$ and $\eta_2$ of the relative velocities, and the difference between two times of collisions $\tau$ (see Proposition 4 in \cite{DoVeAr1}). As a consequence, we expect that understanding the asymptotic behaviour of these three real variables only would allow to understand the full dynamical system \eqref{EQUATSS2.2IterationSystm1}. Besides, under the assumption that the norms of the relative velocities are bounded from below by a positive constant (after a sufficiently large number of collisions), we will now see that the variable $\tau$ is also estimated by the normal components $\eta_1$ and $\eta_2$. Indeed, the most important consequence of the assumption $\vert W_{2,n}^\perp \vert \geq \overline{w} > 0$ is that in this case the Zhou-Kadanoff parameter $\zeta_n$ satisfies
\begin{align}
\zeta_n \lessgtr \frac{d_n}{\eta_{2,n}^2} \hspace{5mm} \text{as } n\rightarrow +\infty.
\end{align}
Therefore, since in order to have a collapse this parameter has always to remain below $1$, we deduce that
\begin{align}
\label{EQUATSS3.3CompaAsymp_d__3}
d_n = \mathcal{O}(\eta_{2,n}^2) \hspace{5mm} \text{as } n\rightarrow +\infty.
\end{align}
Considering now the formula \eqref{EQUATSS2.2TempsColliTauV3} of the time of collision $\tau_n$, we find:
\begin{align}
\label{EQUATSS3.3CompaAsympTauV2}
\tau_n \leq \frac{2d_n}{\vert \eta_{2,n} \vert} = \underbrace{\frac{2d_n}{\eta_{2,n}^2}}_{\text{bounded}} \cdot\ \vert \eta_{2,n} \vert,
\end{align}
(keeping in mind that $\eta_2$ is negative) so that in particular we have the estimate:
\begin{align}
\tau_n = \mathcal{O}(\eta_{2,n}).
\end{align}
\noindent
This observation and its consequences are summarized in the following proposition.

\begin{propo}[Asymptotic comparison of the variables in the collapsing regime, assuming that the relative velocities do not vanish]
\label{PROPOSS3.3AsympCompa__II_}
Let $r \in\ ]0,1[$ be a positive real number smaller than $1$, and let us consider a system of three inelastic particles \circled{0}, \circled{1} and \circled{2} evolving according to the $r$-inelastic hard sphere flow on a time interval $[0,\tau^*[$, and undergoing a nearly-linear inelastic collapse at time $\tau^* > 0$. Let us assume that there exists $\overline{w}>0$ such that $\vert W_{2,n} \vert \geq \overline{w}$ for all $n$ large enough.\\
Let us denote by $t_{\varphi(n)}$ the different times when a collision of type \circled{0}-\circled{1} takes place, immediately followed at times $t_{\varphi(n)+1} = t_{\varphi(n)}+\tau_n$ by a collision of type \circled{0}-\circled{2}. At time $t_{\varphi(n)}$, let us denote by $d_n$ the distance between the particles \circled{0} and \circled{2}, by $\omega_{1,n}$ and $\omega_{2,n}$ the respective angular parameters between the pairs \circled{0} and \circled{1}, and \circled{0} and \circled{2} respectively, by $\eta_{1,n}$ and $\eta_{2,n}$ the respective normal components of the relative velocities, and by $W_{1,n}^\perp$ and $W_{2,n}^\perp$ the respective tangential components of the relative velocities of the pairs \circled{0} and \circled{1}, and \circled{0} and \circled{2} respectively. At time $t_{\varphi(n)+1} = t_{\varphi(n)}+\tau_n$, let us denote by $d'_n$ the distance between the particles \circled{0} and \circled{1}, and by $\omega_{1,n}'$ and $\omega_{2,n}'$ the angular parameters, by $\eta_{1,n}'$ and $\eta_{2,n}'$ the normal components, and by $(W_{1,n}^\perp)'$ and $(W_{2,n}^\perp)'$ the respective tangential components of the relative components of the pairs \circled{0} and \circled{1}, and \circled{0} and \circled{2} respectively.\\
Then, we have the following asymptotic relations, as $n \rightarrow +\infty$:
\begin{align}
\label{EQUATSS3.3PROPOComp2Tau_d}
\tau_n = \mathcal{O}(\eta_{2,n}), \hspace{3mm} \text{and} \hspace{5mm} d_n = \mathcal{O}(\eta_{2,n}^2),
\end{align}
\begin{align}
\label{EQUATSS3.3PROPOComp2_d'__}
d'_n = \mathcal{O}\left(\eta_{1,n}\tau_n\right),
\end{align}
\begin{align}
\label{EQUATSS3.3PROPOComp2_eta'}
\eta_{1,n}'  = \mathcal{O}\left(\eta_{1,n},\eta_{2,n}\right), \hspace{3mm} \text{and} \hspace{5mm} \eta_{2,n}' = \mathcal{O}(\eta_{2,n}),
\end{align}
\begin{align}
\label{EQUATSS3.3PROPOComp2omga'}
\left\vert \omega_{1,n}'-\omega_{1,n} \right\vert = \mathcal{O} \left( \tau_n  \right), \hspace{3mm} \text{and} \hspace{5mm} \left\vert \omega_{2,n}'-\omega_{2,n} \right\vert = \mathcal{O}\left(\tau_n\right)
\end{align}
and finally:
\begin{align}
\label{EQUATSS3.3PROPOComp2PerpD} 
\left\vert (W_{1,n}^\perp)'-W_{1,n}^\perp \right\vert &= \mathcal{O} \left(\eta_{2,n} \right), \hspace{3mm} \text{and} \hspace{5mm} \left\vert (W_{2,n}^\perp)'-W_{2,n}^\perp \right\vert = \mathcal{O}\left(\tau_n\right).
\end{align}
\end{propo}

\begin{remar} Even under the assumption on the positive bound from below for the norm of the relative velocities, one cannot conclude a priori that $\tau$ will be eventually negligible with respect to the normal components: it might perfectly remain of the same order asymptotically. Actually, as we shall see in Section \ref{SSSCT4.1.4EquilZhKadInsta}, the case of the unstable Zhou-Kadanoff equilibrium provides the example of a situation in which $\tau$ and the normal components are asymptotically equivalent.\\
Note also that, in order to establish the estimates of Proposition \ref{PROPOSS3.3AsympCompa__II_}, we needed \emph{only} the information that the \emph{second} relative velocity $\vert W_2 \vert$ does not vanish, which is not surprising, since the equations \eqref{EQUATSS2.2IterationSystm1} are describing a collision of type \circled{0}-\circled{2}. Of course, when a collision of type \circled{0}-\circled{1} is described, the assumption that $\vert W_1 \vert$ does not vanish becomes necessary to obtain the results of Proposition \ref{PROPOSS3.3AsympCompa__II_}. In the end, in the case of the nearly-linear collapse, the two conditions on the norms of $\vert W_1 \vert$ and $\vert W_2 \vert$ are necessary in order to deduce the asymptotics of Proposition \ref{PROPOSS3.3AsympCompa__II_} for all the collisions.
\end{remar}

\section{Reduction to a two-dimensional dynamical systems in the general case}
\label{SECTION__4Reduc2DimeDynSy}

The purpose of this section is to investigate the full dynamical system \eqref{EQUATSS2.2IterationSystm1}, taking into account the non-linearities. Our approach presents new features in the sense that we use not only the general results of Section \ref{SSECTIO2.3ProprGenerColla}, but also the asymptotics of the variables given in Section \ref{SSECTIO2.4CompoAsympVaria} in order to identify in a rigorous manner the leading order terms. In particular, in Section \ref{SSECTIO2.4CompoAsympVaria} we observed that all the vanishing variables of the full dynamical system \eqref{EQUATSS2.2IterationSystm1} can be estimated only in terms of $\eta_1$, $\eta_2$ and $\tau$.\\
As it was already observed by Zhou and Kadanoff (\cite{ZhKa996}, see also Section 3.3.4 of the companion paper \cite{DoVeAr1}), if we can show that $\tau \ll \eta$ asymptotically in the collapsing regime, then the dynamics of the system of the three inelastic particles becomes simpler: the evolution of $\eta_1$ and $\eta_2$ is, to the first order approximation, reduced to a linear mapping, so that considering the evolution law $\eta_1/(-\eta_2)$ yields a one dimensional-iteration, that in turns leads the evolution of all the other variables. We can then rephrase the situation as follows: the problem here consists in checking if the asymptotic relation $\tau \ll \eta$ holds, at least for a set of full measure. As we shall see, we expect actually that the regime $\tau \ll \eta$ is the only one associated to a set of initial configurations of positive measure, that is, we believe it is the only stable regime.\\
\newline
In order to study rigorously when this regime takes place, we will proceed to a reduction of the system \eqref{EQUATSS2.2IterationSystm1} to the two dimensional simpler system. An example of such a two-dimensional reduction has been performed by Zhou and Kadanoff in \cite{ZhKa996}, that considered the two variables $k_n = u_{n+1}^c / u_n^c$ (that is, in our notations: $k_n = \eta_1'/\eta_2$) and $\alpha_n = - t_n/u_n^c$ ($- \tau/\eta_2$ in our notations). However, the variables used in such a reduction are defined using two different time steps of the dynamical system. In the present, we find more convenient to work with variables that correspond to the same time step, and we will consider instead the two following variables:
\begin{align}
\varphi_1 = \frac{\eta_1}{(-\eta_2)} = \frac{\eta_1}{\vert \eta_2 \vert},
\end{align}
and, motivated by \eqref{EQUATSS3.3CompaAsympTauV2}:
\begin{align}
\varphi_2 = \frac{d}{\eta_2^2} \cdotp
\end{align}

\begin{remar}
This two-dimensional reduction is entirely different from considering only the dynamics of the pair of the normal components $\eta_1,\eta_2$, because the evolution of $\eta_1$ and $\eta_2$ is encoded into the evolution of the single variable $\varphi_1 = \eta_1/\vert \eta_2 \vert$. Indeed, as we shall see below, since the evolution of these components is expected to be essentially linear in the Zhou-Kadanoff regime, one can reduce the evolution of these normal components into a single real variable, which is our variable $\varphi_1$. Such a study was originally performed by Zhou and Kadanoff, and it provided the results of Section \ref{SSSCT2.3.2ResulZhou&Kadan}.
\end{remar}




\noindent
Let us start with identifying the leading order terms of the equations of evolution \eqref{EQUATSS2.2IterationSystm1}, keeping in mind the results of the asymptotic behaviours of the different variables of the system obtained in Section \ref{SSECTIO2.4CompoAsympVaria}. Observe that we will seek these leading order terms in full generality, in the sense that we will not rely on any assumption on the asymptotic behaviour of the Zhou-Kadanoff parameter $\zeta$.\\
Let us also note that we consider $-\eta_2 = \vert \eta_2 \vert$ and not $\eta_2$, because the normal component of the relative velocity $W_2$ between the two particles \circled{0} and \circled{2} has to be negative when a collision of type \circled{0}-\circled{1} just took place in order to have a collision in the future between these particles. In particular, when $(-\eta_2) < 0$, the inelastic collapse cannot take place, since the particles would immediately stop colliding.
\noindent


\subsection{Leading order terms of the dynamical system I: the general case, for a single collision}
Back to the equations \eqref{EQUATSS2.2IterationSystm1}, we obtain immediately:
\begin{align*}
\eta_1' = \left(1 + \mathcal{O}(d')\right) \left(\eta_1 + \tau \vert W_1 \vert^2\right) - \frac{(1+r)}{2} \left(\cos\overline{\theta} + o(1)\right) \left( \eta_2 + \tau \vert W_2 \vert^2 + \mathcal{O}\left(d\eta_2\right)\right),
\end{align*}
where $\overline{\theta}$ is the angle, at the limiting time of the collapse, between the particles \circled{1} and \circled{2} measured from the particle \circled{0}. As a consequence, we obtain the following leading order terms (in the neighbourhood of the collapse) of the law of evolution for the normal component of the first relative velocity $\eta_1$:
\begin{align}
\label{EQUATSS4.1LineaEta_1Vers1}
\eta_1' = \eta_1 + \tau \vert W_1 \vert^2 - \frac{(1+r)}{2} (-\cos\overline{\theta})\left((-\eta_2) - \tau \vert W_2 \vert^2\right).
\end{align}
The leading order terms for the normal component of the second relative velocity $W_2$ are obtained in the same way:
\begin{align*}
\eta_2' = -r \eta_2 + \mathcal{O}\left(d\eta_2\right) - r \tau \vert W_2 \vert^2,
\end{align*}
so that the linearization for this component writes:
\begin{align}
\label{EQUATSS4.1LineaEta_2Vers1}
\eta_2' = r(-\eta_2) - r \tau \vert W_2 \vert^2.
\end{align}
Now, we can take advantage of the explicit expression \eqref{EQUATSS2.2TempsColliTauV3} of the collision time $\tau$. In particular, we have:
\begin{align*}
\tau \vert W_2 \vert^2 = (1+d)(-\eta_2) \left[ 1 - \sqrt{1-\zeta} \right] = (-\eta_2) \left[ 1 - \sqrt{1-\zeta} \right] + \mathcal{O}\left( d\eta_2 \right).
\end{align*}
Therefore, the linearized equations \eqref{EQUATSS4.1LineaEta_1Vers1} and \eqref{EQUATSS4.1LineaEta_2Vers1} can be rewritten as:
\begin{align}
\label{EQUATSS4.1LineaEta_1Vers2}
\eta_1' = \eta_1 + \frac{\vert W_1 \vert^2}{\vert W_2 \vert^2} \left[1 - \sqrt{1-\zeta}\right](-\eta_2) - \frac{(1+r)}{2} (-\cos\overline{\theta})\sqrt{1-\zeta}(-\eta_2),
\end{align}
and
\begin{align}
\label{EQUATSS4.1LineaEta_2Vers2}
\eta_2' = r(-\eta_2)\sqrt{1-\zeta}.
\end{align}
From \eqref{EQUATSS4.1LineaEta_1Vers1} and \eqref{EQUATSS4.1LineaEta_2Vers1} we observe that the dynamics of the normal components $\eta_1$ and $\eta_2$ is entirely prescribed (to the first order approximation), and we can reduce the dynamics to a one-dimensional iteration by studying the evolution of the ratio:
\begin{align}
\varphi_1 = \eta_1/(-\eta_2).
\end{align}
Proceeding that way, we would recover an iteration that is equivalent to the system obtained by Zhou and Kadanoff in \cite{ZhKa996}, in the so-called \emph{flat surface approximation}.\\
The last variable we need to understand completely the first order approximation of the full dynamical system \eqref{EQUATSS2.2IterationSystm1} is $\tau$. According to the discussion in the previous Section \ref{SSECTIO2.4CompoAsympVaria}, we will consider the evolution of the ratio
\begin{align}
\varphi_2 = d/\eta_2^2 \cdotp
\end{align}
In order to obtain the iteration defining the dynamical system of the nearly-linear inelastic collapse, we need now to use the assumption on the order of the collisions. After the collision \circled{0}-\circled{2}, we need to exchange the roles of the particles \circled{1} and \circled{2} in the evolution equations \eqref{EQUATSS2.2IterationSystm1} in order to describe a collision of type \circled{0}-\circled{1}. Then, we obtain the following approximated equations for the normal components:
\begin{align}
\label{EQUATSS4.1LineaEta_1Vers3}
\eta_1' = r(-\eta_2)\sqrt{1-\zeta}.
\end{align}
and
\begin{align}
\label{EQUATSS4.1LineaEta_2Vers3}
\eta_2' = \eta_1 + \frac{\vert W_1 \vert^2}{\vert W_2 \vert^2} \left[1 - \sqrt{1-\zeta}\right](-\eta_2) - \frac{(1+r)}{2} (-\cos\overline{\theta})\sqrt{1-\zeta}(-\eta_2),
\end{align}
obtained simply by exchanging $\eta_1'$ and $\eta_2'$ in the previous equations \eqref{EQUATSS4.1LineaEta_1Vers2} and \eqref{EQUATSS4.1LineaEta_2Vers2}. Note that, if the consecutive iterations describe a collapsing system of particles, the sign of $\eta_2$ has to be negative for each iteration. Concerning the evolution of $\varphi_2$, we have:
\begin{align*}
\varphi_2' = \frac{d'}{(-\eta_2')^2}.
\end{align*}
The leading order terms of $d'$, corresponding to the distance between the particles \circled{0} and \circled{1} when the pair \circled{0}-\circled{2} is colliding at the time $\tau$, are immediately obtained from the expression \eqref{EQUATSS2.2IterationSystm1} of the full dynamical system:
\begin{align*}
d' &= \eta_1 \tau + \frac{1}{2} \vert W_1 \vert^2 \tau^2 \\
&= \frac{1}{\vert W_2 \vert^2} \eta_1 (-\eta_2) \left[1-\sqrt{1-\zeta}\right] + \frac{1}{2} \frac{\vert W_1 \vert^2}{\vert W_2 \vert^4} (-\eta_2)^2 \left[1 - \sqrt{1-\zeta}\right]^2.
\end{align*}
In the end, we find:
\begin{align}
\label{SECT6EquatDefinPhi_1Vers1}
\varphi_1' &= \frac{\eta_1'}{(-\eta_2')} = \frac{r(-\eta_2)\sqrt{1-\zeta}}{\displaystyle{\frac{(1+r)}{2} (-\cos\overline{\theta})\sqrt{1-\zeta}(-\eta_2) - \frac{\vert W_1 \vert^2}{\vert W_2 \vert^2} \left[1 - \sqrt{1-\zeta}\right](-\eta_2) -\eta_1}} \nonumber\\
&= \frac{r\sqrt{1-\zeta}}{\displaystyle{\left(\frac{(1+r)}{2} (-\cos\overline{\theta}) + \frac{\vert W_1 \vert^2}{\vert W_2 \vert^2}\right)\sqrt{1-\zeta} - \varphi_1 - \frac{\vert W_1 \vert^2}{\vert W_2 \vert^2}}},
\end{align}
and
\begin{align}
\label{SECT6EquatDefinPhi_2Vers1}
\varphi_2' &= \frac{d'}{(-\eta_2')^2} = \frac{\left[1-\sqrt{1-\zeta}\right]}{\vert W_2 \vert^2} \frac{ \displaystyle{\eta_1(-\eta_2) + \frac{\vert W_1 \vert^2}{2\vert W_2 \vert^2} \left[1-\sqrt{1-\zeta}\right] (-\eta_2)^2 }}{ \left(\displaystyle{\left(\frac{(1+r)}{2} (-\cos\overline{\theta}) + \frac{\vert W_1 \vert^2}{\vert W_2 \vert^2}\right)\sqrt{1-\zeta}(-\eta_2) - \eta_1 - \frac{\vert W_1 \vert^2}{\vert W_2 \vert^2}(-\eta_2)} \right)^2} \nonumber \\
&= \frac{\left[1-\sqrt{1-\zeta}\right]}{\vert W_2 \vert^2} \frac{ \displaystyle{\varphi_1 + \frac{\vert W_1 \vert^2}{2\vert W_2 \vert^2} \left[1-\sqrt{1-\zeta}\right] }}{ \left(\displaystyle{\left(\frac{(1+r)}{2} (-\cos\overline{\theta}) + \frac{\vert W_1 \vert^2}{\vert W_2 \vert^2}\right)\sqrt{1-\zeta} - \varphi_1 - \frac{\vert W_1 \vert^2}{\vert W_2 \vert^2}} \right)^2} \cdotp
\end{align}
Let us emphasize that we can obtain a closed form of the system, only in terms of the variables $\varphi_1$ and $\varphi_2$, noting that the definition \eqref{EQUATSS2.2DefinParZK_zeta} of $\zeta$ implies the following approximation in the neighbourhood of a collapse:
\begin{align}
\sqrt{1 - \zeta} \simeq \sqrt{1 - 2 \vert W_2 \vert^2 \varphi_2^2}.
\end{align}
Note that we implicitly assumed that the norms $\vert W_1 \vert$ and $\vert W_2 \vert$ are converging towards positive constants, and we used an abuse of notation since such limiting constants are denoted also by $\vert W_1 \vert$ and $\vert W_2 \vert$ in the expressions \eqref{SECT6EquatDefinPhi_1Vers1} and \eqref{SECT6EquatDefinPhi_2Vers1} above.\\
Introducing in the end the following last notations:
\begin{align}
\alpha = \frac{(1+r)}{2}(-\cos\overline{\theta}), \hspace{3mm} \vert W_1 \vert^2 \xrightarrow[]{} a > 0, \hspace{3mm} \vert W_2 \vert^2 \xrightarrow[]{} b > 0,
\end{align}
and finally
\begin{align}
\frac{ \vert W_2 \vert^2}{\vert W_1 \vert^2} \xrightarrow[]{} \frac{b}{a} = T,
\end{align}
we can rewrite the iterations for the variables $\varphi_1$ and $\varphi_2$ in their final form.

\begin{defin}[Leading order reduction of the dynamical system \eqref{EQUATSS2.2IterationSystm1} in the case of the nearly-linear inelastic collapse for non vanishing tangential velocities]
\label{DEFINSS4.1Itera_2dim1Coll}
Let $r \in\, ]0,1[$ be a positive number strictly smaller than $1$, let $a$ and $b$ be two strictly positive real numbers, and let $\overline{\theta}$ be an angle in $[\pi/2,\pi[$.\\
We define the \emph{leading order reduction of the dynamical system \eqref{EQUATSS2.2IterationSystm1} in the case of the nearly-linear inelastic collapse for non vanishing tangential velocities}, or, in short, the \emph{one-collision mapping}, as the function:
\begin{subequations}
\label{EQUATSS4.1Itera_2dim1Coll}
\begin{empheq}[left=\empheqlbrace]{align}
\varphi_1' &= \frac{r\sqrt{1-2b\varphi_2}}{\displaystyle{\left( \alpha + \frac{1}{T}\right)\sqrt{1-2b\varphi_2} - \varphi_1 - \frac{1}{T}}}, \label{EQUATSS4.1Itera_2dim1Col1}\\
\varphi_2' &= \frac{\left[1-\sqrt{1-2b\varphi_2}\right]}{b} \frac{ \displaystyle{\left(\varphi_1 + \frac{1}{2T} \left[1-\sqrt{1-2b\varphi_2}\right] \right)}}{ \left(\displaystyle{\left(\alpha + \frac{1}{T}\right)\sqrt{1-2b\varphi_2} - \varphi_1 - \frac{1}{T}} \right)^2}, \label{EQUATSS4.1Itera_2dim1Col2}
\end{empheq}
\end{subequations}
where $\alpha = \frac{(1+r)}{2}(-\cos\overline{\theta})$ and $T = b/a$.
\end{defin}

\noindent
The two equations \eqref{EQUATSS4.1Itera_2dim1Coll} of Definition \ref{DEFINSS4.1Itera_2dim1Coll} are central in our study.

\begin{remar}
Note that the mapping \eqref{EQUATSS4.1Itera_2dim1Coll} corresponds to the effect of a single collision, that is meant to represent \emph{only} a collision between \circled{0} and \circled{2} (assuming that the norm of the relative velocities $\vert W_1 \vert$ and $\vert W_2 \vert$ respectively converge towards $a,b > 0$). Indeed, in order to describe the next collision, between the pair \circled{0}-\circled{1}, the constants $a$ and $b$ have to be exchanged, and $T$ has to be changed into $1/T$. As a consequence, we cannot iterate directly (and naively) the mapping defined in \eqref{EQUATSS4.1Itera_2dim1Coll}.
\end{remar}


\subsection{Leading order terms of the dynamical system II: the two-iteration mapping}

The expression \eqref{EQUATSS4.1Itera_2dim1Coll} describes a first collision, between the particles \circled{0} and \circled{2}. Let us now write the expression corresponding to the second collision, between \circled{0} and \circled{1}. Composing this new mapping with \eqref{EQUATSS4.1Itera_2dim1Coll}, we will obtain the dynamical system associated to the consecutive pair of collisions \circled{0}-\circled{2}, \circled{0}-\circled{1}, that we will finally iterate to describe the evolution of the leading order terms of the full dynamical system \eqref{EQUATSS2.2IterationSystm1} in the case of a nearly-linear inelastic collapse.\\
For the sake of completeness, and in order for the reader to realize the complexity of the dynamical system, we will write explicitly the composition of the two mappings. However, the details of the computations are postponed to the Appendix (see Section \ref{APPENSC___DetaiCalcuDbIte}).\\ 
We assume that a collision of type \circled{0}-\circled{2} already took place, and the configuration obtained after such a collision is $\left(\varphi_1',\varphi_2'\right)$. The leading order terms of the evolution of $\left(\varphi_1',\varphi_2'\right)$ during the next collision, of type \circled{0}-\circled{1}, is:
\begin{align}
\label{EQUATSS4.1IteraIntermPhi1}
\varphi_1'' = \frac{r \sqrt{1 - 2a \varphi_2'}}{\displaystyle{\left(\alpha + T\right) \sqrt{1-2a\varphi_2'} - \varphi_1' - T}} \cdotp
\end{align}
We can then write:
\begin{align*}
1-2a\varphi_2' = 1 - 2 \frac{b}{T} \varphi_2' = 1 - \frac{2}{T} \left[ 1 - \sqrt{1-2b\varphi_2} \right]\left( \frac{\varphi_1 + \frac{1}{2T}\left[1 - \sqrt{1-2b\varphi_2}\right]}{\displaystyle{\left( \left(\alpha + \frac{1}{T}\right)\sqrt{1-2b\varphi_2} - \varphi_1 - \frac{1}{T} \right)^2}} \right).
\end{align*}
Reducing to the same denominator we have first:
\begin{align*}
\left( \left(\alpha + \frac{1}{T}\right)\sqrt{1-2b\varphi_2} - \varphi_1 - \frac{1}{T} \right)^2 =& \left(\alpha + \frac{1}{T}\right)^2\left(1-2b\varphi_2\right) + \varphi_1^2 + \frac{1}{T^2} \\
&- 2\left(\alpha+\frac{1}{T}\right)\varphi_1\sqrt{1-2b\varphi_2} - 2 \left(\alpha+\frac{1}{T}\right)\frac{1}{T}\sqrt{1-2b\varphi_2} + 2 \frac{\varphi_1}{T} \cdotp
\end{align*}
Let us reuse again the notation $\zeta = 2b\varphi_2$ in order to make the computation clearer. We have then:
\begin{align}
\label{EQUATSS4.1DbIteTerme1-2aP}
1-2a\varphi_2' = \left( \left(\alpha + \frac{1}{T}\right)\sqrt{1-2b\varphi_2} - \varphi_1 - \frac{1}{T} \right)^{-2} \Bigg[ \alpha \left(\alpha+\frac{2}{T}\right)(1-\zeta) - 2\alpha\left( \varphi_1 + \frac{1}{T} \right)\sqrt{1-\zeta} + \varphi_1^2 \Bigg].
\end{align}
Therefore we find for the first variable:
\begin{align}
\label{SECT6EquatDefinPhi_1Secde} 
\varphi_1'' &= \frac{r \sqrt{ \displaystyle{ \alpha \left(\alpha+\frac{2}{T}\right)(1-\zeta) - 2\alpha\left( \varphi_1 + \frac{1}{T} \right)\sqrt{1-\zeta} + \varphi_1^2 }}}{ \displaystyle{(\alpha + T) \sqrt{ \displaystyle{ \alpha \left(\alpha+\frac{2}{T}\right)(1-\zeta) - 2\alpha\left( \varphi_1 + \frac{1}{T} \right)\sqrt{1-\zeta} + \varphi_1^2 }} - (T+\varphi_1') \left\vert \left(\alpha+\frac{1}{T}\right)\sqrt{1-\zeta} - \varphi_1 - \frac{1}{T} \right\vert } },
\end{align}
where $\zeta$ denotes the Zhou-Kadanoff parameter, defined as $\zeta = 2b\varphi_2$, and $\varphi_1'$ is given by the expression \eqref{EQUATSS4.1Itera_2dim1Col1}.\\
In the same way, considering the second variable, we have:
\begin{align}
\label{EQUATSS4.1IteraIntermPhi2}
\varphi_2'' = \frac{\left[1-\sqrt{1-2(b/T)\varphi_2'}\right]}{b/T} \frac{\displaystyle{\left( \varphi_1' + \frac{T}{2}  \left[ 1 - \sqrt{1-2(b/T)\varphi_2'} \right] \right)}}{\left( \displaystyle{(\alpha+T)\sqrt{1-2(b/T)\varphi_2'} - T - \varphi_1'} \right)^2}\cdotp
\end{align}
If we denote by $R_1$ the square root:
\begin{align}
\label{EQUATSS4.1RacineCarrNotat}
R_1 = \sqrt{\alpha \left(\alpha+\frac{2}{T}\right)(1-\zeta) - 2\alpha\left( \varphi_1 + \frac{1}{T} \right)\sqrt{1-\zeta} + \varphi_1^2},
\end{align}
and using again \eqref{EQUATSS4.1Itera_2dim1Col1}, we obtain:
\begin{align*}
\varphi_2'' =& \frac{T}{b} \frac{ \displaystyle{ \left\vert \left(\alpha+\frac{1}{T}\right)\sqrt{1-\zeta} - \varphi_1 - \frac{1}{T} \right\vert - R_1 } }{ \left( \displaystyle{ (\alpha+T) R_1 - (T+\varphi_1')\left\vert \left(\alpha+\frac{1}{T}\right)\sqrt{1-\zeta} - \varphi_1 - \frac{1}{T} \right\vert } \right)^2 } \\
&\hspace{60mm} \times \left( \left( \varphi_1' + \frac{T}{2} \right) \left\vert \left(\alpha+\frac{1}{T}\right)\sqrt{1-\zeta} - \varphi_1 - \frac{1}{T} \right\vert - \frac{T}{2} R_1 \right) \cdotp
\end{align*}
Let us recall that the term
\begin{align*}
\left( \alpha + \frac{1}{T} \right) \sqrt{1-\zeta} - \varphi_1 - \frac{1}{T}
\end{align*}
corresponds to
\begin{align*}
\frac{(-\eta_2)'}{\eta_1'} r\sqrt{1-\zeta},
\end{align*}
hence it has to be positive if we consider a system that experiences infinitely many collisions. Therefore, using the formula \eqref{EQUATSS4.1Itera_2dim1Col1} that gives $\varphi_1'$, it is possible to provide an explicit expression for $\varphi_2''$ in terms of the initial variables $\varphi_1$ and $\varphi_2$. A first step provides:
\begin{align}
\label{SECT6EquatDefinPhi_2Secde}
\varphi_2'' =& \frac{T}{b} \frac{ \displaystyle{ \left(\alpha+\frac{1}{T}\right)\sqrt{1-\zeta} - \varphi_1 - \frac{1}{T} - R_1 } }{ \left( \displaystyle{ (\alpha+T) R_1 - (T+\varphi_1')\left\vert \left(\alpha+\frac{1}{T}\right)\sqrt{1-\zeta} - \varphi_1 - \frac{1}{T} \right\vert } \right)^2 } \nonumber\\
&\hspace{50mm} \times \left( r\sqrt{1-\zeta} + \frac{T}{2} \left( \left(\alpha+\frac{1}{T}\right)\sqrt{1-\zeta} - \varphi_1 - \frac{1}{T} - R_1 \right) \right),
\end{align}
with $\zeta = 2b\varphi_2$ and $R_1$ defined in \eqref{EQUATSS4.1RacineCarrNotat}.\\
At this point, it should be clear that studying the evolution of the dynamical system of three inelastic particles experiencing an inelastic collapse is in no mean an easy task, as keeping only the leading order terms of the evolution law \eqref{EQUATSS2.2IterationSystm1} leads to the expressions above. Let us introduce the final mapping we will study in this work, before moving to its properties. 

\begin{defin}[Leading order reduction of the double iteration of the dynamical system \eqref{EQUATSS2.2IterationSystm1} in the case of the nearly-linear inelastic collapse for non vanishing tangential velocities]
\label{DEFINSS4.1Itera_2dim2Coll}
Let $r \in\, ]0,1[$ be a positive number strictly smaller than $1$, let $a$ and $b$ be two strictly positive real numbers, and let $\overline{\theta}$ be an angle in $[\pi/2,\pi[$.\\
We define the \emph{leading order reduction of the double iteration of the dynamical system \eqref{EQUATSS2.2IterationSystm1} in the case of the nearly-linear inelastic collapse for non vanishing tangential velocities}, or, in short, the \emph{two-collision mapping}, as the the function:
\begin{subequations}
\label{EQUATSS4.1Itera_2dim2Coll}
\begin{empheq}[left=\empheqlbrace]{align}
\varphi_1'' &= \frac{r \sqrt{1 - 2a \varphi_2'}}{\displaystyle{\left(\alpha + T\right) \sqrt{1-2a\varphi_2'} - \varphi_1' - T}} , \label{EQUATSS4.1Itera_2dim2Col1}\\
\varphi_2'' &= \frac{\left[1-\sqrt{1-2(b/T)\varphi_2'}\right]}{b/T} \frac{\displaystyle{\left( \varphi_1' + \frac{T}{2}  \left[ 1 - \sqrt{1-2(b/T)\varphi_2'} \right] \right)}}{\left( \displaystyle{(\alpha+T)\sqrt{1-2(b/T)\varphi_2'} - T - \varphi_1'} \right)^2}, \label{EQUATSS4.1Itera_2dim2Col2}
\end{empheq}
\end{subequations}
where $\alpha = \frac{(1+r)}{2}(-\cos\overline{\theta})$ and $T = b/a$, and where $\varphi_1'$ and $\varphi_2'$ are respectively defined in \eqref{EQUATSS4.1Itera_2dim1Col1} and \eqref{EQUATSS4.1Itera_2dim1Col2}.
\end{defin}
\noindent
The expressions \eqref{EQUATSS4.1Itera_2dim1Coll}, \eqref{EQUATSS4.1Itera_2dim2Coll} describes the leading order evolution of the variables $\left(\varphi_1,\varphi_2\right)$ around the collapse regime when a pair of consecutive collisions \circled{0}-\circled{2}, \circled{0}-\circled{1} takes place, assuming that the square of the norms of the relative velocities $\vert W_1 \vert$ and $\vert W_2 \vert$ converge respectively to $a$ and $b$. Iterating infinitely many times these expressions will provide the asymptotic behaviour of the particle system.\\
\newline
Let us emphasize here the fact that the derivation of the expressions \eqref{EQUATSS4.1Itera_2dim1Coll}, \eqref{EQUATSS4.1Itera_2dim2Coll} can only be considered as formal so far. Nevertheless, we obtained these expressions relying on the rigorous results of Proposition \ref{PROPOSS3.3AsympCompa__II_}, that indicate which terms of the full dynamical system \eqref{EQUATSS2.2IterationSystm1} can be neglected. To make this derivation fully rigorous, it would be necessary to show that the orbits of the full dynamical system \eqref{EQUATSS2.2IterationSystm1} are approximated by the associated orbits of the two-dimensional reduction \eqref{EQUATSS4.1Itera_2dim1Coll}, \eqref{EQUATSS4.1Itera_2dim2Coll}, with an arbitrary precision. We hope to be able to obtain this approximation result in a future work.\\
We believe that the two-collision mapping \eqref{EQUATSS4.1Itera_2dim1Coll}, \eqref{EQUATSS4.1Itera_2dim2Coll} encodes the main features of the evolution of the dynamical system \eqref{EQUATSS2.2IterationSystm1} when a nearly-linear collapse takes place. The next sections will be devoted to the careful study of this two-collision mapping.

\subsection{Preliminary properties of the two-collision mapping \eqref{EQUATSS4.1Itera_2dim2Coll}}
\label{SSSCT4.1.3PremiProprDbIte} 

First of all, let us notice that the iteration \eqref{EQUATSS4.1Itera_2dim1Coll}-\eqref{EQUATSS4.1Itera_2dim2Coll} makes sense only for initial data $\left(\varphi_1,\varphi_2\right)$ such that $1 - 2b\varphi_2 \geq 0$, or again:
\begin{align*}
\varphi_2 \leq \frac{1}{2b}.
\end{align*}
On the other hand, the ``physical'' domain of study of the iterations is such that the two variables $\varphi_1$ and $\varphi_2$ remain both non-negative at each iteration. In particular, assuming that $\varphi_2$ is non negative, we have always $\sqrt{1 - 2b\varphi_2} \leq 1$, and so if $\varphi_1 > \alpha$, then we have
\begin{align*}
\left(\alpha+\frac{1}{T}\right) \sqrt{1-2b\varphi_2} - \varphi_1 - \frac{1}{T} < \alpha + \frac{1}{T} - \alpha - \frac{1}{T} < 0,
\end{align*}
so that after a single collision (that is, computing $\varphi_1'$ according to \eqref{EQUATSS4.1Itera_2dim1Coll}) we have $\varphi_1' < 0$. Therefore we see that the relevant domain of study is contained in the rectangle $[0,\alpha]\times[0,1/2b]$.\\
In the same way, one can provide more information considering the second iteration: assuming that $\varphi_2'$ is non negative, if $\varphi_1' > \alpha$, $\varphi_1''$ would be negative. But this inequality on $\varphi_1'$ can be rewritten as:
\begin{align*}
\alpha \varphi_1 \leq \left(\alpha + \frac{1}{T} - r\right)\sqrt{1-2b\varphi_2} - \frac{\alpha}{T},
\end{align*}
or again, since $\varphi_1 + 1/T > \varphi_1 \geq 0$:
\begin{align*}
\alpha^2 \left(\varphi_1 + \frac{1}{T} \right)^2 \leq \left(\alpha+\frac{1}{T}-r\right)^2 (1-2b\varphi_2),
\end{align*}
which provides
\begin{align*}
2b\left(\alpha+\frac{1}{T}-r\right)^2 \varphi_2 \leq \left(\alpha+\frac{1}{T}-r\right)^2 - \alpha^2\left(\varphi_1+\frac{1}{T}\right)^2,
\end{align*}
or again
\begin{align*}
\varphi_2 \leq \frac{1}{2b} - \frac{\alpha^2\left(\varphi_1+\frac{1}{T}\right)^2}{2b\left(\alpha+\frac{1}{T}-r\right)^2},
\end{align*}
which is in the end strictly more stringent than $\varphi_2 \leq 1/2b$.

\begin{remar}
One could also study the condition $\varphi_2' \leq T/(2b)$, although it is more complicated to use. In general, we can refine the domain of definition by considering arbitrary many iterations of \eqref{EQUATSS4.1Itera_2dim1Coll} and \eqref{EQUATSS4.1Itera_2dim2Coll}, together with the constraints on the signs of $\varphi_1$ and $\varphi_2$ after such iterations.
\end{remar}
\noindent
We can gather our observations in the following proposition.

\begin{propo}[Relevant domain for the two-collision mapping \eqref{EQUATSS4.1Itera_2dim2Coll}]
\label{PROPOSS4.1DomaiPositDbIte} 
Let us consider $\varphi_1,\varphi_2 \geq 0$. In order to have the first and second iterates \eqref{EQUATSS4.1Itera_2dim1Coll}, \eqref{EQUATSS4.1Itera_2dim2Coll} preserving the non-negativity of the variables $\varphi_1$ and $\varphi_2$, the initial data $\varphi_{1,0},\varphi_{2,0}$ have to be chosen such that
\begin{align}
0 \leq \varphi_{1,0} \leq \alpha
\end{align}
and
\begin{align}
\varphi_{2,0} \leq \frac{1}{2b} - \frac{\alpha^2\left(\varphi_1+\frac{1}{T}\right)^2}{2b\left(\alpha+\frac{1}{T}-r\right)^2}
\end{align}
\end{propo}
\noindent
In addition, we can notice that the regime $\varphi_2 = 0$ is an invariant set for the two-collision mapping. Indeed, taking $\varphi_2 = 0$ in the first iteration \eqref{EQUATSS4.1Itera_2dim1Coll} provides immediately:
\begin{align*}
\varphi_2'(\varphi_1,0) = \frac{\left[1-\sqrt{1-0}\right]}{b} \frac{ \left( \displaystyle{\varphi_1 + \frac{1}{2T} \left[1-\sqrt{1-0}\right] } \right) }{ \left(\displaystyle{\left(\alpha + \frac{1}{T}\right)\sqrt{1-0} - \varphi_1 - \frac{1}{T}} \right)^2} = 0,
\end{align*}
and then in the same way
\begin{align*}
\varphi_2''(\varphi_1',0) = 0.
\end{align*}
\noindent
The region $\varphi_2 = 0$ corresponds to the particular regime studied by Zhou and Kadanoff in \cite{ZhKa996}, that the authors call the ``flat surface approximation''. The question whether this regime is attracting or not is central in our study, for it would imply the ZK-regime. Let us consider the dynamics of the double iteration in the region $\varphi_2 = 0$.\\
Back to the first iteration \eqref{EQUATSS4.1Itera_2dim1Coll}, in the case $\varphi_2 = 0$ we find
\begin{align*}
\varphi_1'(\varphi_1,0) = \frac{r}{\displaystyle{\left( \alpha + \frac{1}{T}\right) - \varphi_1 - \frac{1}{T}}} = \frac{r}{\alpha - \varphi_1},
\end{align*}
so that
\begin{align}
\label{EQUATSS4.1IteraPhi_1Phi20}
\varphi_1''(\varphi_1',0) &= \frac{r}{\alpha - \varphi_1'} \nonumber\\
&= \frac{r}{\displaystyle{\alpha - \frac{r}{\alpha - \varphi_1}}} \nonumber\\
&= \frac{r(\alpha-\varphi_1)}{\alpha(\alpha-\varphi_1)-r} \cdotp
\end{align}
Therefore, the one-dimensional iteration \eqref{EQUATSS4.1IteraPhi_1Phi20} describes the evolution of the system in the invariant region $\varphi_2=0$.\\
The equation of the fixed points of \eqref{EQUATSS4.1IteraPhi_1Phi20} is:
\begin{align*}
\alpha\varphi_1(\alpha-\varphi_1) -r\varphi_1 = r(\alpha-\varphi_1),
\end{align*}
which simplifies into
\begin{align*}
\varphi_1(\alpha-\varphi_1) = r
\end{align*}
considering $\alpha \neq 0$ (which is equivalent to $\overline{\theta} \neq \pi/2$, where $\overline{\theta}$ describes the limiting angle, at the final time of the collapse, between the particles \circled{1} and \circled{2}, measured from the center of \circled{0}).\\
Therefore, in the case when $\alpha^2 - 4r\geq 0$, the double iteration has two fixed points (that coincide in the threshold case $\alpha^2 = 4r$):
\begin{align}
\varphi_1^\pm = \frac{\alpha \pm \sqrt{\alpha^2 -4r}}{2} \cdotp
\end{align}
Let us notice that the one-dimensional iteration \eqref{EQUATSS4.1IteraPhi_1Phi20} can be rewritten as:
\begin{align*}
\varphi \mapsto f(\varphi_1),
\end{align*}
with
\begin{align*}
f(x) = \frac{r}{\alpha} + \frac{r^2/\alpha}{\alpha(\alpha-x) - r} \cdotp
\end{align*}
The graph of $f$ is an hyperbola, with two increasing branches, separated by the vertical asymptote $x = \alpha^{-1}\left(\alpha^2-r\right) = \alpha - r/\alpha$. At $x=\pm\infty$, the curve presents also an horizontal asymptote $y = r/\alpha$. Since its derivative is given by the expression
\begin{align*}
f'(x) = \frac{r^2}{\left(\alpha(\alpha-x) - r\right)^2},
\end{align*}
the slope of the tangents to the curve of $f$ at $x$ is equal to $1$ if and only if $x = \alpha - 2r/\alpha$ or $x = \alpha$ (one point solving this condition per branch). Since in particular we have $f(\alpha) = 0$ (corresponding to the second branch, associated to $x > \alpha - r/\alpha$), we deduce that the intersections between the line $y=x$ and the graph of $f$, if they exist, have to belong to the first branch of the graph of $f$ (associated to $x < \alpha - r/\alpha$). Therefore, considering the convexity of this first branch, we deduce that
\begin{align*}
\left\vert f'(\varphi_1^-) \right\vert < 1 \text{   and    } \left\vert f'(\varphi_1^+) \right\vert > 1.
\end{align*}
As a consequence, if $\alpha^2 - 4r > 0$, ensuring that there exist two distinct fixed points, we can deduce that the larger $\varphi_1^+$ is always unstable, whereas the smaller $\varphi_1^-$ is always stable (keeping in mind that we consider here the dynamics only in the invariant region $\varphi_2 = 0$).\\
Let us also note that Proposition \ref{PROPOSS4.1DomaiPositDbIte} provided a necessary condition on the initial datum $\left(\varphi_1,\varphi_2\right)$ such that the two-collision mapping \eqref{EQUATSS4.1Itera_2dim2Coll} preserves the positivity. In the invariant region $\varphi_2 = 0$, we find a more stringent upper bound on $\varphi_1$: if $\varphi_1 \in [\alpha-r/\alpha,\alpha[$, $f(\varphi_1)$ is either not defined, or negative. On the other hand, the interval $[0,\alpha-r/\alpha[$ in invariant for the mapping $f$. We can then summarize these observations in the following statement.

\begin{propo}[Dynamics of the two-collision mapping \eqref{EQUATSS4.1Itera_2dim2Coll} in the invariant region $\varphi_2 = 0$.]
We consider the dynamical system defined as the iterations of the two-collision mapping \eqref{EQUATSS4.1Itera_2dim2Coll} in the invariant region $[0,\alpha-r/\alpha[\times\{0\} \subset\{\varphi_2 = 0\}$.\\
Then, the dynamics of the two-collision mapping reduces to the following iteration $\varphi_1 \mapsto \varphi_1''$ in the region $\varphi_2 = 0$, where:
\begin{align}
\label{EQUATSS4.1DoublIteraPhi20}
\varphi_1'' = \frac{r}{\alpha} + \frac{r^2/\alpha}{\alpha(\alpha-\varphi_1) - r} \cdotp
\end{align}
Depending on the value of $\alpha^2-r$, the dynamics in the invariant region $\varphi_2 = 0$ is as follows.
\begin{itemize}
\item If $\alpha^2 - r < 0$, the mapping \eqref{EQUATSS4.1DoublIteraPhi20} has no fixed point. Depending on the value $\varphi_{1,0}$ of the initial datum (chosen in $\mathbb{R})$, the sequence $\left(\varphi_{1,n}\right)_n$ obtained after recursive applications of the double iteration \eqref{EQUATSS4.1DoublIteraPhi20} either reaches the value $\alpha - r/\alpha$ after a finite number of iterations, after which the sequence is not well-defined, or it alternates between infinite number of iterations inside $]-\infty,\alpha-r/\alpha[$ (where the sequence is increasing) and $]\alpha-r/\alpha,+\infty[$ (where the sequence is decreasing).
\item If $\alpha^2 - r = 0$, the double iteration \eqref{EQUATSS4.1DoublIteraPhi20} has a single fixed point $\varphi_1^\pm$, given by the expression
\begin{align}
\varphi_1^\pm = \frac{\alpha}{2} \cdotp
\end{align}
Depending on the value $\varphi_{1,0}$ of the initial datum (chosen in $\mathbb{R})$, the sequence $\left(\varphi_{1,n}\right)_n$ obtained after recursive applications of the double iteration \eqref{EQUATSS4.1DoublIteraPhi20} either reaches the value $\alpha - r/\alpha$ after a finite number of iterations, after which the sequence is not well-defined, or it converges towards $\varphi_1^\pm$.
\item If $\alpha^2 - r < 0$, the double iteration \eqref{EQUATSS4.1DoublIteraPhi20} has two fixed points, given by the expression
\begin{align}
\varphi_1^\pm = \frac{\alpha \pm \sqrt{\alpha^2 -4r}}{2},
\end{align}
$\varphi_1^-$ (the smaller) being stable, whereas $\varphi_1^+$ (the larger) being unstable.\\
If the value $\varphi_{1,0}$ of the initial datum belongs to $[0,\varphi_1^+[$, the sequence $\left(\varphi_{1,n}\right)_n$ obtained after recursive applications of the double iteration \eqref{EQUATSS4.1DoublIteraPhi20} converges towards $\varphi_1^-$.
\end{itemize}
\end{propo}

\noindent
The region $\varphi_2 = 0$ is remarkable, since the dynamics can be studied completely in this region. Therefore, a natural idea is to determine what information can be obtained in the neighbourhood of such a region.\\
\noindent
In particular, we would like to determine the values of $\varphi_1$ for which the orbits of the two-collision mapping starting from $\left(\varphi_1,\varphi_2\right)$, with $\varphi_2$ small enough, are eventually converging towards the invariant region $\varphi_2$. We know that in the region $\varphi_2=0$ the two-collision mapping the fixed points $\varphi_1^\pm$ of the dynamics restricted to the invariant region $\varphi_2 = 0$. Indeed, we know that, when these two fixed points exist, only $\varphi_-$ is stable, and that its basin of attraction corresponds to $[0,\varphi_1^+[$. Therefore, we aim in particular to compare the region for which \eqref{SECT6EquatCondiStabiPh2=0} holds true with this basin of attraction.\\
\noindent
We obtain the following result.

\begin{theor}[Dynamics of the two-collision mapping \eqref{EQUATSS4.1Itera_2dim2Coll} close to the invariant region $\varphi_2 = 0$.]
\label{THEORSS4.1DynamDblIteP2<1}
Let $r$ be a strictly positive number, smaller than a certain positive constant $r_0 < 1$. Then, the dynamics of the double iteration \eqref{EQUATSS4.1Itera_2dim1Coll}, \eqref{EQUATSS4.1Itera_2dim2Coll} has two fixed points $0 < \varphi_1^- < \varphi_1^+$ in the invariant region $\varphi_2=0$. In addition, there exists a positive number $\overline{\varphi}_1 = \overline{\varphi}_1(r)$, and a curve $\varphi_2 = \psi (\varphi_1)$ such that if the initial datum $(\varphi_{1,0},\varphi_{2,0})$ of the double iteration satisfies $0 \leq \varphi_{1,0} < \overline{\varphi}_1$ and $0 \leq \varphi_{2,0} < \psi(\varphi_1)$, then the orbit of the dynamics produced by the double iteration starting from $(\varphi_{1,0},\varphi_{2,0})$ is globally well-posed, and converges exponentially fast towards $(\varphi_1^-,0)$.
\end{theor}

\begin{proof}
The two-collision mapping describing the dynamics of the particles is obtained by composing two simpler mappings, the first one being described by the expressions \eqref{EQUATSS4.1Itera_2dim1Coll}, while the second mapping is obtained in \eqref{EQUATSS4.1Itera_2dim2Coll} from the expression of the first, except that $b$ is replaced by $a=b/T$, and $T$ by $1/T$. In particular, the first mapping can be written as
\begin{align*}
\left( \varphi_1,\varphi_2 \right) \mapsto F^1 \left(\varphi_1,\varphi_2\right) = \left( f^1_1(\varphi_1,\varphi_2), f^1_2(\varphi_1,\varphi_2) \right),
\end{align*}
(the superscript indicates which collision we are considering, while the subscript describes the coordinates) with
\begin{align}
\label{EQUATSS4.1Itera1CollCoord}
\left\{
\begin{array}{rl}
f^1_1(x,y) &= \frac{r \sqrt{1-2by}}{ \left( \alpha+\frac{1}{T}\right)\sqrt{1-2by} - x - \frac{1}{T}} \\
f^1_2(x,y) &= \frac{\left(1 - \sqrt{1-2by}\right) \left( x + \frac{1}{2T}\left(1-\sqrt{1-2by}\right) \right)}{\left( \left( \alpha + \frac{1}{T} \right) \sqrt{1-2by} - x - \frac{1}{T} \right)^2} \cdotp
\end{array}
\right.
\end{align}
The mapping \eqref{EQUATSS4.1Itera1CollCoord} is acting on $\mathbb{R}^2$. It turns out that the Jacobian matrix of the mapping $\left(\varphi_1,\varphi_2\right) \mapsto \left(f^1_1(\varphi_1,\varphi_2),f^1_2(\varphi_1,\varphi_2)\right)$ of the first collision can be computed explicitly, providing a result of a reasonable complexity. The computations are presented in the Appendix (see Section \ref{APPENSS___CalcuDbIteJacob}). As for the first coordinate $f^1_1$, we obtain:
\begin{align*}
\partial_x f_1(x,y) = \frac{r\sqrt{1-2by}}{\left(\left( \alpha + \frac{1}{T} \right) \sqrt{1-2by} - x - \frac{1}{T} \right)^2},
\end{align*}
as well as
\begin{align}
\partial_y f_1(x,y) &= \frac{rb \left( x + \frac{1}{T} \right)}{ \sqrt{1-2by} \left( \left( \alpha + \frac{1}{T} \right) \sqrt{1-2by} - x - \frac{1}{T} \right)^2 } \cdotp
\end{align}
Concerning the second coordinate $f^1_2$, we find first
\begin{align}
\partial_x f^1_2(x,y) &= \frac{\left( 1 - \sqrt{1-2by} \right) \left[ \alpha \sqrt{1-2by} + x \right]}{ \left( \left( \alpha + \frac{1}{T} \right) \sqrt{1-2by} - x - \frac{1}{T} \right)^3 } \cdotp
\end{align}
For the second derivative of the last coordinate, we obtain finally:
\begin{align}
\partial_y f^1_2(x,y) &= \frac{ x }{\sqrt{1-2by} \left( \left(\alpha + \frac{1}{T}\right)\sqrt{1-2by} - x - \frac{1}{T} \right)^2} \nonumber\\
&\hspace{10mm} +\frac{(1-\sqrt{1-2by})\left( (\alpha+1/T) \left( \sqrt{1-2by} + 2x + (1-\sqrt{1-2by})/T \right) - x - 1/T \right) }{\sqrt{1-2by} \left( \left(\alpha + \frac{1}{T}\right)\sqrt{1-2by} - x - \frac{1}{T} \right)^3} \cdotp
\end{align}
\noindent
With the help of the partial derivatives of the coordinates, we can now easily compute the Jacobian matrix of the two-collision mapping in the region $\varphi_2 = 0$. Indeed, first, we observe that the expressions of the partial derivatives take a particular simple form when we choose an initial datum in the region $\varphi_2 = 0$. We get:
\begin{align*}
\partial_x f^1_1(x,0) = \frac{r}{\left( \left(\alpha+\frac{1}{T}\right) - x - \frac{1}{T} \right)^2} = \frac{r}{(\alpha-x)^2},\\
\partial_y f^1_1(x,0) = \frac{rb\left( x + \frac{1}{T} \right)}{\left( \left(\alpha+\frac{1}{T}\right) - x - \frac{1}{T} \right)^2} = \frac{rb\left( x + \frac{1}{T} \right)}{(\alpha-x)^2} \cdotp
\end{align*}
In the same way, we obtain
\begin{align*}
\partial_x f^1_2(x,0) &= 0,\\
\partial_y f^1_2(x,0) &= \frac{\left[ x\left( \left(\alpha + \frac{1}{T}\right) - x - \frac{1}{T} \right) \right]}{ \left( \left(\alpha + \frac{1}{T}\right) -x - \frac{1}{T} \right)^3} = \frac{x(\alpha-x)}{(\alpha-x)^3} = \frac{x}{(\alpha-x)^2} \cdotp
\end{align*}
Hence the following expression for the Jacobian matrix of the mapping of the first collision, in the region $\varphi_2 = 0$:
\begin{align*}
\text{Jac} F^1(\varphi_1,0) =
\begin{pmatrix}
\frac{r}{\left( \alpha - \varphi_1 \right)^2} & \frac{rb \left(\varphi_1 + \frac{1}{T}\right)}{ \left( \alpha - \varphi_1 \right)^2} \\
0 & \frac{\varphi_1}{\left( \alpha - \varphi_1 \right)^2}
\end{pmatrix} \cdotp
\end{align*}
Now, writing the two-collision mapping, describing the two collisions, as $F^2 \circ F^1$, where $F^2$ is the mapping of the second collision, obtained from the expression of $F^1$ by replacing $b$ by $b/T$ and $T$ by $1/T$, we find by the chain rule:
\begin{align*}
\text{Jac} \left( F^2 \circ F^1 \right) (\varphi_1,0) = \text{Jac} F^2 (\varphi_1',0) \cdot \text{Jac} F^1(\varphi_1,0),
\end{align*}
where $(\varphi_1',0)$ is the image of $(\varphi_1,0)$ by the first collision mapping $F^1$. In particular, we have
\begin{align*}
\varphi_1' = \frac{r}{\left(\alpha + \frac{1}{T}\right) - \varphi_1 - \frac{1}{T}} = \frac{r}{\alpha - \varphi_1},
\end{align*}
and
\begin{align*}
\text{Jac} F^2 (\varphi_1',0) =
\begin{pmatrix}
\frac{r}{(\alpha - \varphi_1')^2} & \frac{\frac{rb}{T}\left(\varphi_1' + T\right)}{(\alpha - \varphi_1')^2} \\
0 & \frac{\varphi_1'}{(\alpha - \varphi_1')^2}
\end{pmatrix} \cdotp
\end{align*}
Since we have
\begin{align*}
\alpha - \varphi_1' = \alpha - \frac{r}{\alpha - \varphi_1} = \frac{\alpha(\alpha-\varphi_1)-r}{\alpha - \varphi_1} = \frac{\alpha^2 - r - \alpha\varphi_1}{\alpha - \varphi_1},
\end{align*}
we obtain in the end
\begin{align*}
\text{Jac} \left( F^2 \circ F^1 \right) (\varphi_1,0) &=
\begin{pmatrix}
\frac{r}{(\alpha - \varphi_1')^2} & \frac{\frac{rb}{T}\left(\varphi_1' + T\right)}{(\alpha - \varphi_1')^2} \\
0 & \frac{\varphi_1'}{(\alpha - \varphi_1')^2}
\end{pmatrix}
\cdotp
\begin{pmatrix}
\frac{r}{\left( \alpha - \varphi_1 \right)^2} & \frac{rb \left(\varphi_1 + \frac{1}{T}\right)}{ \left( \alpha - \varphi_1 \right)^2} \\
0 & \frac{\varphi_1}{\left( \alpha - \varphi_1 \right)^2}
\end{pmatrix} \\
& = \begin{pmatrix}
\frac{r^2}{(\alpha-\varphi_1')^2(\alpha-\varphi_1)^2} & \frac{r^2b \left(\varphi_1+\frac{1}{T}\right)}{(\alpha-\varphi_1')^2(\alpha-\varphi_1)^2} + \frac{r b \varphi_1 (\varphi_1'+T)}{T(\alpha-\varphi_1')^2(\alpha-\varphi_1)^2} \\
0 & \frac{\varphi_1 \varphi_1'}{(\alpha - \varphi_1')^2(\alpha - \varphi_1)^2}
\end{pmatrix} \\
&= \begin{pmatrix}
\frac{r^2}{(\alpha^2 - r - \alpha \varphi_1)^2} & \frac{r b \left( rT(\varphi_1+1/T) + \varphi_1\left(\frac{r}{\alpha-\varphi_1} + T\right) \right)}{T(\alpha^2 - r - \alpha \varphi_1)^2} \\
0 & \frac{r \varphi_1}{(\alpha-\varphi_1)(\alpha^2 - r - \alpha\varphi_1)^2}
\end{pmatrix} \cdotp
\end{align*}
The Jacobian matrix of the two-collision mapping is upper triangular in the region $\varphi_2=0$. Therefore we can immediately read the eigenvalues of this matrix on its diagonal. The first eigenvalue, associated to the eigenspace spanned by $\left(1,0\right)$, is
\begin{align*}
\lambda_1 = \frac{r^2}{(\alpha^2 - r - \alpha \varphi_1)^2} \cdotp
\end{align*}
Only the second eigenvalue, which is
\begin{align*}
\lambda_2 = \frac{r \varphi_1}{(\alpha-\varphi_1)(\alpha^2 - r - \alpha\varphi_1)^2},
\end{align*}
encodes some two-dimensional phenomena. The expression of $\lambda_2$ has two singularities, respectively $\alpha - r/\alpha$ and $\alpha$, both of them being larger than the fixed points of the mapping of the double iteration restricted to $\varphi_2 = 0$. Therefore, on our domain of study ($0 \leq \varphi_1 \leq \alpha - r/\alpha$), this second eigenvalue is never singular.\\
Considering the image of a small perturbation of a point lying in the invariant region $\varphi_2 = 0$, that is, considering the image of $(\varphi_1,h)$, where $h$ is small, we find:
\begin{align*}
F(\varphi_1,h) = F\left( (\varphi_1,0) + (0,h) \right) \simeq \underbrace{F(\varphi_1,0)}_{= (\varphi_1'',0)} + h \left( \partial_y F_1(\varphi_1,0) , \partial_y F_2(\varphi_1,0) \right),
\end{align*}
so that
\begin{align*}
F_2(\varphi_1,h) \simeq \lambda_2 h.
\end{align*}
We see therefore that the invariant region $\varphi_2 = 0$ is attracting (at least in a small neighbourhood) if $\left\vert \lambda_2 \right\vert < 1$. Such a condition can be rewritten as
\begin{align*}
\frac{r \varphi_1}{(\alpha-\varphi_1)(\alpha^2-r-\alpha \varphi_1)^2} < 1,
\end{align*}
or again
\begin{align}
\label{SECT6EquatCondiStabiPh2=0}¨
(\varphi_1-\alpha)(\alpha^2-r-\alpha\varphi_1)^2 + r \varphi_1 < 0,
\end{align}
when $\varphi_1 < \alpha$. The left hand side of \eqref{SECT6EquatCondiStabiPh2=0} can be rewritten as the polynomial:
\begin{align}
\label{SECT6EquatPolynStabiPh2=0}
P_r(x) = \alpha^2 x^3 - \alpha (3\alpha^2 - 2r) x^2 + \left( (\alpha^2-r)(3\alpha^2-r) + r\right) x - \alpha(\alpha^2-r)^2.
\end{align}
First, let us notice that $P_r(0) = - \alpha(\alpha^2-r)^2 < 0$. In particular, the region $\varphi_2 = 0$ is always attracting in the plane $\left(\varphi_1,\varphi_2\right)$, at least for $\varphi_1$ and $\varphi_2$ small enough.\\
Now, let us notice that the explicit comparison between the fixed points and the condition \eqref{SECT6EquatCondiStabiPh2=0} is technically difficult, due to the third degree polynomial $P_r$. However, we can state the following results. The derivative of $P_r$ writes:
\begin{align*}
P'_r(x) &= \frac{\dd}{\dd x}\left[ (x-\alpha) (\alpha^2 - r - \alpha x)^2 + r x \right] \\
&= (\alpha^2 - r - \alpha x)^2 - 2 \alpha (x-\alpha) (\alpha^2 - r - \alpha x) + r,
\end{align*}
so the minimum of the derivative $P'_r$ lies at
\begin{align*}
x = \frac{2 \alpha(3\alpha^2 - 2r)}{6 \alpha^2} = \alpha - \frac{2r}{3 \alpha} \cdotp
\end{align*}
Let us note that when $r=0$, we find $P_r(\alpha) = \alpha^5 - 3 \alpha^5 + 3 \alpha^5 - \alpha^5 = 0$ as well as $P_r'(\alpha) = 0$. We can in addition develop the expression of $P_r'\left(\alpha-\frac{2r}{3\alpha}\right)$ when $r$ is small, and we find:
\begin{align*}
P'(\alpha - 2r/(3\alpha)) = \left( -r +2r/3 \right)^2 - 2 \alpha \left(-2r/(3\alpha)\right) ( -r +2r/3 ) + r = r + o(r) \hspace{3mm} \text{as } r\rightarrow 0.
\end{align*}
Therefore, we can deduce that for all $r$ small enough, the derivative $P'_r$ is always positive. In particular, the function $P_r$ has a unique zero, that is positive for all $r$ small enough. In the same way, since $\varphi_1^+ = \frac{\alpha + \sqrt{\alpha^2 - 4r}}{2} = \frac{\alpha}{2} \left( 1 + \sqrt{1 + \frac{4r}{\alpha^2}}\right) = \alpha(1 + r/\alpha^2) + o(r) = \alpha + r/\alpha + o(r)$ as $r \rightarrow 0$, we find
\begin{align*}
P_r(\varphi_1^+) &= \alpha^2 (\alpha + r/\alpha)^3 - \alpha(3\alpha^2 - 2r)(\alpha + r/\alpha)^2 + \left((\alpha^2 -r)(3\alpha^2-r) + r\right)(\alpha + r/\alpha) - \alpha(\alpha^4 - 2r\alpha^2) + o(r)\\
&= \alpha^2(\alpha^3 + 3r\alpha) - \alpha(3\alpha^2-2r)(\alpha^2 + 2r) + \left( (3\alpha^4 - 4\alpha^2 r) + r \right)(\alpha + r/\alpha) - \alpha^5 + 2 \alpha^3 r + o(r)\\
&= \left[ \alpha^3(3-6+2+3-4+2) + \alpha \right] r + o(r)\\
&= \alpha r + o(r)
\end{align*}
as $r \rightarrow 0$. So, for all $r$ small enough, $P_r$ evaluated at the second fixed point $\varphi_1^+$ is positive.\\
As a consequence, for all $r$ small enough, there exists $\overline{\varphi}_1 \in\, ]\varphi_1^-,\varphi_1^+[$ such that in the region $\varphi_1 < \overline{\varphi}_1$, the stable line $\varphi_2 = 0$ is attracting, up to choose $\varphi_2$ small enough. As a conclusion, the Zhou-Kadanoff regime is attracting, up to choose an initial datum close enough to the invariant region $\varphi_2 = 0$, and far enough from the unstable fixed point.\\
The previous computations provides in the end the stability of the region $\varphi_2=0$ in the plane $\left(\varphi_1,\varphi_2\right)$. In addition, we deduce directly the long-time behaviour of the orbits of the dynamics of the double iteration, starting close enough to the invariant region $\varphi_2 = 0$, together with a lower bound on the rate of convergence: the orbits converge exponentially fast towards the unique stable fixed point $(\varphi_1^-,0)$, since the second coordinate of any orbit behaves essentially like a geometric sequence, where the ratio is given by the second eigenvalue $\lambda_2$ of the Jacobian matrix of the double iteration. On the other hand, the dynamics of the first coordinate is essentially governed by the dynamics restricted to the invariant region $\varphi_2 = 0$, which converges also exponentially fast towards the stable fixed point $\varphi_1$. The proof of Theorem \ref{THEORSS4.1DynamDblIteP2<1} is complete.
\end{proof}

\begin{remar}
The result of Theorem \ref{THEORSS4.1DynamDblIteP2<1} is a perturbative result, in the sense that, from the dynamics on the line $\varphi_2=0$, we can deduce the dynamics of the two-collision mapping \eqref{EQUATSS4.1Itera_2dim2Coll} locally around this line.
\end{remar}

\begin{remar}
Along the proof of Theorem \ref{THEORSS4.1DynamDblIteP2<1}, we obtained that $\partial_x F_2$ is zero. Such a result corresponds to the fact that a small perturbation in the first coordinate has no first order effect on the second coordinate of the image, when one starts from an initial datum in the region $\varphi_2 = 0$. We recover that the region $\varphi_2 = 0$ is invariant under the action of the two-collision mapping.\\
In addition, we obtained that the first eigenvalue $\lambda_1$ of $\text{Jac}\left(F^2\circ F^1\right)$ is always positive (when it is well-defined), and has an increasing value as $\varphi_1$ increases from $0$, until its singularity, reached at $\varphi_1 = \alpha - r/\alpha$, this singularity being also the singular point of the dynamics when restricted to the invariant region $\varphi_2 = 0$. For $\varphi_1 = 0$, the first eigenvalue $\lambda_1$ is equal to
\begin{align*}
\frac{r^2}{\alpha^2-r}\cdotp
\end{align*}
When restricted to the invariant region $\varphi_2 = 0$, the first coordinate $F^1$ governs entirely the dynamics, in particular $\lambda_1$ plays a decisive role in the stability of the fixed points.
\end{remar}
\noindent
Determining the basin of attraction of the invariant region $\varphi_2 = 0$, that is, ultimately, of the stable fixed point $(\varphi_1^+,0)$, turns out to be a delicate question. However, this question is central, for studying completely this basin of attraction is equivalent to determine completely the domain of validity of the Zhou-Kadanoff regime in terms of initial data in the plane $(\varphi_1,\varphi_2)$. Let us conclude this section with results of numerical simulations of the two-collision mapping that suggest a geometrical description of the basin of attraction of $\varphi_2 = 0$.

\subsection{Qualitative behaviour of the two-collision mapping: observations and conjectures}
\label{SSSCT4.1.4NumerConjectuDCM} 

We collected in the previous section the rigorous results we obtained concerning the two-collision mapping \eqref{EQUATSS4.1Itera_2dim2Coll}. Nevertheless, such results hold only in a neighbourhood of the invariant region $\varphi_2=0$. Therefore, it would be valuable to investigate its behaviour, beyond the local results around the stable line $\varphi_2 = 0$. In this section, we will report numerical observations we made concerning the orbits of the two-collision mapping. We will consider the case when the two-collision mapping \eqref{EQUATSS4.1Itera_2dim2Coll} has two different fixed points (one of them being stable) exist on the stable manifold $\varphi_2 = 0$, as described in Theorem \ref{THEORSS4.1DynamDblIteP2<1}.

\paragraph{The domain of global well-posedness of the orbits.}
We discussed already that the possibility to compute an orbit of the two-collision mapping \eqref{EQUATSS4.1Itera_2dim2Coll} requires four conditions: in order to compute the first collision \eqref{EQUATSS4.1Itera_2dim1Coll}, we need (besides the physically motivated constraints $\varphi_1',\varphi_1'',\varphi_2',\varphi_2'' \geq 0$):
\begin{itemize}
\item $1-2b\varphi_2 \geq 0$, so that $\sqrt{1-2b\varphi_2}$ is well-defined,
\item and $\left(\alpha+1/T\right)\sqrt{1-2b\varphi_2} - \varphi_1 - 1/T \neq 0$ in order to be able to divide by this quantity.
\end{itemize}
Similarly, for the second collision, we need:
\begin{align*}
1-2a\varphi_2' \geq 0 \hspace{5mm} \text{and} \hspace{5mm} (\alpha+T)\sqrt{1-2a\varphi_2'} - \varphi_1' - T \neq 0.
\end{align*}
These four conditions can be rewritten only in terms of the variables $\varphi_1$, $\varphi_2$. In order to apply further iterates of the two-collision mapping, we would find four new conditions for each iteration. A natural question is to describe the set $(\varphi_1,\varphi_2)$ satisftying all these conditions together. To do so, let us denote by $\mathcal{D} \subset \left(\mathbb{R}_+\right)^2$ the largest invariant set in the quadrant $\varphi_1 \geq 0,\varphi_2 \geq 0$ under the action of the two-collision mapping, that we will call the \emph{domain of the two-collision mapping}.

\paragraph{Description of the numerical simulations.} Determining $\mathcal{D}$ seems a difficult problem, but we performed numerical simulations in order to obtain some insight about this set. The code was designed to store and print the initial data generating orbits well-defined for more than a certain number $N$ of collisions (which corresponds to fulfill the four conditions of all of the $N$ first iterates of \eqref{EQUATSS4.1Itera_2dim2Coll}). 
At this step, the boundary of the domain $\mathcal{D}$ could have been a very complicated set, as it is sometimes the case with two-dimensional dynamical systems (see for instance \cite{Draz005}). Though, repeated simulations, with various choices of parameters, exhibited all a regular structure for the domain $\mathcal{D}$. Considering finer and finer grids of initial data for the orbits, or longer and longer orbits (that is, with $N$ larger and larger), we observed on the numerical simulations the that the domain $\mathcal{D}$ is not empty (as expected), and it presented a regular boundary (contrary to the Mandelbrot set, for instance). More explicitly, it seems that there exists a point $p_\text{up}$ of coordinates $(0,y_{p_\text{up}})$ in the $\left(\varphi_1,\varphi_2\right)$-plane, and a curve $\mathcal{S}$ through that point, which is the boundary of the domain $\mathcal{D}$, and that is the graph of a function $\Phi$ of the first variable $\varphi_1$, decreasing until it crosses the first axis, through the larger fixed point $p^+ = \left(\varphi_+,0\right)$ of the two-collision mapping on the stable manifold $\varphi_2=0$. We can rephrase our numerical observations as follows: if we denote by $\left(x,\Phi(x)\right)$ the coordinates of the points of the curve $\mathcal{S}$ that we observed, it seems that $0 \leq \varphi_2 < \Phi(\varphi_1)$ if and only if the orbit issued from $\left(\varphi_1,\varphi_2\right)$ is globally well-defined. This curve $\mathcal{S}$, that we clearly observed, and that we pictured in red in Figure \ref{FIGURSS4.1PortrPhasePhi12} below, seems very regular, and can be interpreted as a \emph{separatrix}: below the curve, we have globally well-posed dynamics, and above, the orbits leave $\left(\mathbb{R}_+\right)^2$ or become ill-defined after a finite number of iterations. Below this curve, we represented the domain $\mathcal{D}$ as a dotted surface on Figure \ref{FIGURSS4.1PortrPhasePhi12}.

\paragraph{Behaviour of the orbits.} Once we had a clear idea of where the globally well-defined orbits are located in the $\left(\varphi_1,\varphi_2\right)$-plane, the next natural question is to study their behaviour. We plotted numerically several orbits, and observed the following facts.
\begin{itemize}
\item It seems that all the orbits starting below the separatrix $\mathcal{S}$ eventually converged towards the stable fixed point $p_- = \left(\varphi_-,0\right)$, that lies inside the domain $\mathcal{D}$.
\item In a vertical stripe in $\mathcal{D}$ above $p_-$, the orbits seem to concentrate along a vertical curve, approaching vertically $p_-$, as pictured on Figure \ref{FIGURSS4.1PortrPhasePhi12}.
\item For initial data in $\mathcal{D}$ of the form $\left(\varphi_{1,0},\varphi_{2,0}\right)$, with $\varphi_{1,0} > \varphi_-$ and $\varphi_{2,0}$ small (typically, for initial data in the right bottom corner of $\mathcal{D}$, close to $p_+$), the orbits starting from these points start to move upwards, following the separatrix from below, then bind to the left, until they reach the thin vertical stripe above the stable fixed point $p_-$. Finally, the orbits match with the ones we observed already in this stripe, and converge towards the fixed point $p_-$.
\end{itemize}
Typical orbits are represented in dashed lines on Figure \ref{FIGURSS4.1PortrPhasePhi12}.

\begin{figure}[h!]
    \centering
   \includegraphics[width=8cm, trim = 0cm 0cm 0cm 0cm]{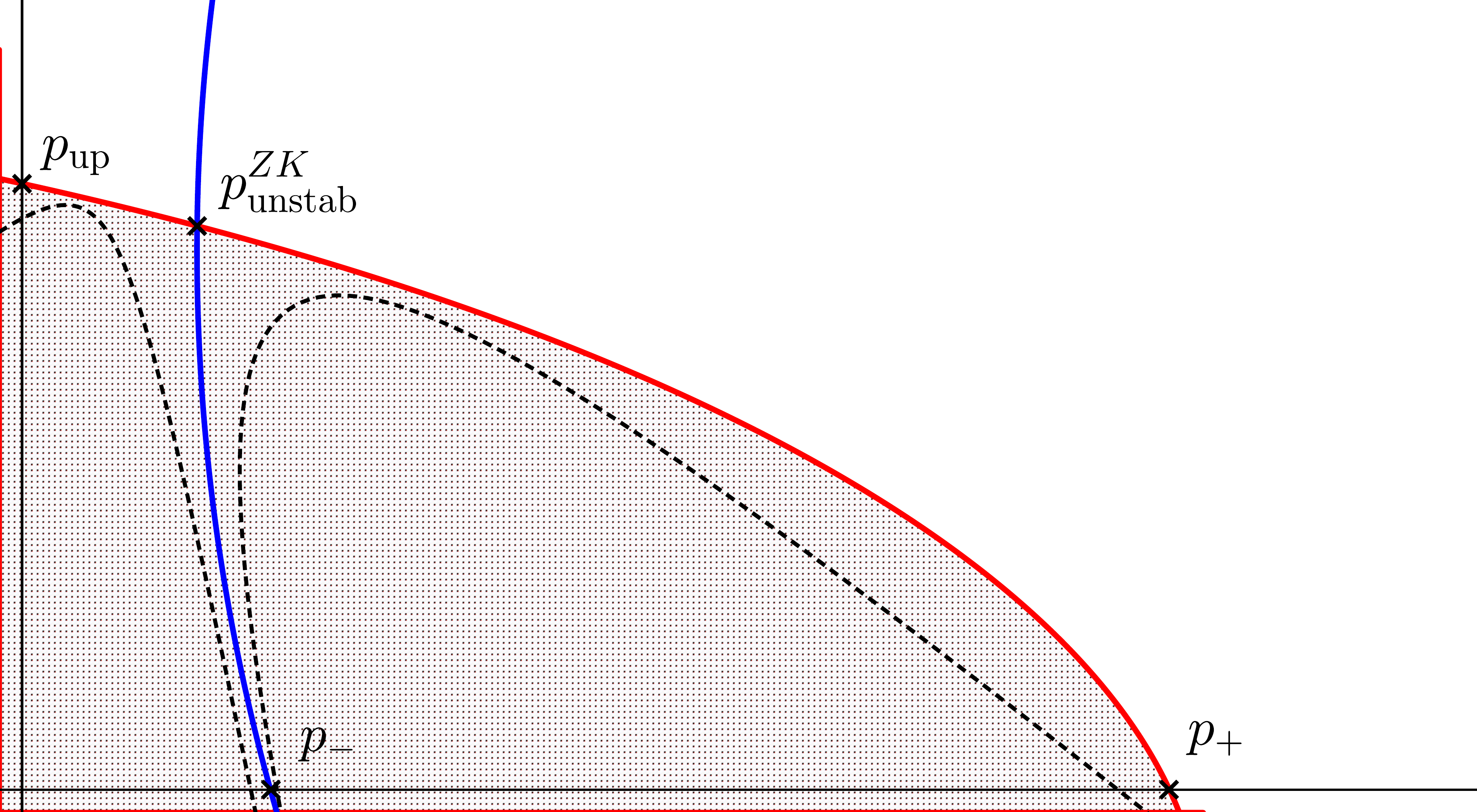}
   \caption{The phase portrait of the two-collision mapping \eqref{EQUATSS4.1Itera_2dim2Coll}, as observed in numerical simulations.}
   \label{FIGURSS4.1PortrPhasePhi12}
\end{figure}

\paragraph{Conjectures.} From the typical orbits we observed, we can infer the dynamics of orbits starting \emph{exactly} on the separatrix $\mathcal{S}$: it seems that such orbits follow the curve, moving upwards, to the left for orbits starting close to $p_+$, and to the right for orbits starting close to $p_\text{up}$, until a point that is approximately above the stable fixed point $p_-$. This seems to mean to the separatrix $\mathcal{S}$ is the stable manifold around the unstable fixed point $p_+$. The point of the separatrix $\mathcal{S}$ approached by the orbits in the first part of their evolution, that we represented and denoted by $p_\text{unstab}^{ZK}$ on Figure \ref{FIGURSS4.1PortrPhasePhi12}, seemingly is an unstable fixed point. Let us observe that an unstable fixed point outside the stable manifold $\varphi_2=0$ was already identified by Zhou and Kadanoff in a particular case, that we discuss in the next section. We believe that the limit of the orbits contained in the separatrix $\mathcal{S}$ is the equilibrium of Zhou and Kadanoff (hence the notation). Finally, it seems that a curve $\mathcal{V}$ links the two fixed points $p_\text{unstab}^{ZK}$ and $p_-$, such that the orbits starting on such a curve remain on it, travelling from the equilibrium of Zhou and Kadanoff towards the stable fixed point $p_-$ (that we proved to be stable, at least locally, in the plane $\left(\varphi_1,\varphi_2\right)$). Such a curve is also represented on Figure \ref{FIGURSS4.1PortrPhasePhi12}: this is the solid line through $p_-$. The orbits computed numerically in the thin stripe above $p_-$ suggest indeed that such a curve $\mathcal{V}$ exists.\\
\newline
In summary, concerning the dynamics of the two-collision mapping \eqref{EQUATSS4.1Itera_2dim2Coll} we postulate the existence of:
\begin{enumerate}
\item the invariant set $\mathcal{S}$ under the action of the two-collision mapping, which is a curve containing the unstable fixed point $p_+$ and a point of the second axis denoted by $p_\text{up}$, parametrized by a decreasing function $\Phi$, where $\left(\varphi_1,\varphi_2\right) \in \mathcal{S} \Leftrightarrow \varphi_2 = \Phi(\varphi_1)$,
\item another invariant set $\mathcal{V}$, which is a curve containing the stable fixed point $p_-$,
\item another unstable equilibrium $p_\text{unstab}^{ZK}$ with strictly positive coordinates, located at the intersection between the curves $\mathcal{S}$ and $\mathcal{V}$, this equilibrium corresponding to the one identified in \cite{ZhKa996} in the particular case $a=b$.
\end{enumerate}
In addition, we postulate the following behaviours:
\begin{enumerate}
\setcounter{enumi}{3}
\item the orbits are globally well-defined if and only if they intersect the space delimited by the first quadrant, and which is below $\mathcal{S}$, and such orbits all converge towards $p_-$,
\item the curve $\mathcal{S}$ is constituted with the unstable fixed point $p_+$, an heteroclinic orbit from $p_+$ to the Zhou-Kadanoff equilibrium $p_\text{unstab}^{ZK}$, this point $p_\text{unstab}^{ZK}$, and another orbit coming from the left, converging also towards the Zhou-Kadanoff equilibrium $p_\text{unstab}^{ZK}$,
\item in the quadrant $\left(\mathbb{R}_+\right)^2$, the curve $\mathcal{V}$ is constituted with the Zhou-Kadanoff equilibrium $p_\text{unstab}^{ZK}$, an heteroclinic orbit from this point to the stable point $p_-$, and this point itself.
\end{enumerate}
If our predictions are correct, the dynamics of the two-collision mapping would be completely understood. In particular, the ZK-regime would be almost surely certain, in the sense that, among the globally well-defined orbits, only the orbits intersecting the curve $\mathcal{S}$ would not satisfy asymptotically the ZK-regime. On the opposite, the orbits intersecting the curve $\mathcal{S}$ (hence, completely contained in that curve) would never satisfy the ZK-regime, although they would be globally well-defined. It seems that the domain $\mathcal{D}$ is extremely regular: the set of initial data leading to orbits asymptotically in the ZK-regime seems open, simply connected. Finally, it seems that the curve $\mathcal{S}$ is the transition between the ill-defined orbits of the two-iteration \eqref{EQUATSS4.1Itera_2dim2Coll} on the one hand, and asymptotically ZK-orbits on the other hand.

\section{Reduced system in the symmetric case ($a=b$)}
\label{SSECTIO4.2ReducSystmSymme}

The complete study of the general dynamical system \eqref{EQUATSS2.2IterationSystm1}, although reduced to the two-dimensional iteration \eqref{EQUATSS4.1Itera_2dim2Coll}, is still out of reach. Nevertheless, as it was already pointed out in \cite{ZhKa996} (see Section V, page 626), there exists a particularly symmetric case in which one can carry on completely the study of the fixed points of the mapping \eqref{EQUATSS4.1Itera_2dim2Coll}. This symmetric case corresponds to the assumption $a=b$, with $a$ and $b$ being the limiting norms of the tangential relative velocities are the same.\\
Let us note that such a case should not be observable in practice, since the constraint $a=b$ for the limiting state of the system is expected to correspond to a set of measure zero for the initial data of the complete dynamical system.\\
\newline
In the case $a=b$, the two different iterations of the reduced two-dimensional dynamical system coincide. Therefore, there is a single mapping to study, which writes as follows.

\begin{subequations}
\label{EQUATSS4.2ReducSystmSymme}
\begin{empheq}[left=\empheqlbrace]{align}
\varphi_1' &= \displaystyle{\frac{r\sqrt{1-2b\varphi_2}}{(\alpha+1)\sqrt{1-2b\varphi_2} - \varphi_1 - 1}}, \label{EQUATSS4.2ReducSystmSym_1}\\
\varphi_2' &= \displaystyle{\frac{\left( 1 - \sqrt{1-2b\varphi_2}\right)}{b}} \cdot \displaystyle{ \frac{ \left( \varphi_1 + \frac{1}{2} \left( 1 - \sqrt{1-2b\varphi_2}\right) \right)}{\left( (\alpha+1) \sqrt{1-2b\varphi_2} - \varphi_1 - 1\right)^2} } \cdotp \label{EQUATSS4.2ReducSystmSym_2}
\end{empheq}
\end{subequations}

\subsection{Unstable Zhou-Kadanoff equilibrium of the symmetric reduced system}
\label{SSSCT4.1.4EquilZhKadInsta}

Let us provide here a complete study of the fixed points of the iteration \eqref{EQUATSS4.2ReducSystmSymme}. The main result of this section, already observed in \cite{ZhKa996}, concerns the existence of a single fixed point outside the lines $\varphi_1=0$ and $\varphi_2 = 0$. Let us state and prove in full details such a result, which is the object of the following theorem.

\begin{theor}[Existence, uniqueness and unstability of the Zhou-Kadanoff equilibrium]
\label{THEORSS4.2PointFixe_ZhKad}
Let us consider the symmetric case $a=b$ of the two-collision mapping \eqref{EQUATSS4.1Itera_2dim2Coll}, which reduces to the one-collision mapping \eqref{EQUATSS4.2ReducSystmSymme}.\\
Then, in addition to the two fixed points $p_- = \left(\varphi_1^-,0\right)$ and $p_+ = \left(\varphi_1^+,0\right)$ on the line $\varphi_2 = 0$ already studied in Section \ref{SSSCT4.1.3PremiProprDbIte},
\begin{itemize}
\item if $r^{1/3} + r^{2/3} \geq \alpha$, there is no other fixed point in the domain $\varphi_1 > 0$, $\varphi_2 \geq 0$,
\item if $r^{1/3} + r^{2/3} < \alpha$, there exists a single fixed point in the domain $\varphi_1 > 0$, $\varphi_2 \geq 0$, denoted by $p_{\text{unstab.}}^{ZK}$. This fixed point is unstable.
\end{itemize}
The fixed point $p_{\text{unstab.}}^{ZK}$ will be called the Zhou-Kadanoff equilibrium. Such a fixed point exists if and only if the smallest fixed point $p_- = \left(\varphi_1^+,0\right)$ is stable in the plane $\left(\varphi_1,\varphi_2\right) \in \left(\mathbb{R}_+\right)^2$.
\end{theor}

\begin{proof}[Proof of Theorem \ref{THEORSS4.2PointFixe_ZhKad}]
We are looking for the points $(x,y) \in\ \left(\mathbb{R}_+\right)^2$ solving:
\begin{align}
x = \frac{r \sqrt{1-2by}}{(\alpha+1)\sqrt{1-2by}-x-1}, \hspace{1cm} y = \frac{1-\sqrt{1-2by}}{b} \cdot \frac{\left( x + \frac{1}{2}\left(1 - \sqrt{1-2by}\right) \right)}{\left( (\alpha+1)\sqrt{1-2by} - x - 1 \right)^2} \cdotp
\end{align}
As we shall see, it will be more convenient to work with an alternative variable instead of $y$. Let us set:
\begin{align}
\label{EQUATSS4.4DefinVaria__w__}
w = 1-\sqrt{1-2by}.
\end{align}
$w$ is small if and only if $y$ is, and we have even $w \underset{y \rightarrow 0}{\sim} by$, so that $w$ and $y$ behave essentially in the same way in the neighbourhood of the Zhou-Kadanoff regime.\\
In terms of the $(x,w)$ variables, the symmetrized iteration \eqref{EQUATSS4.2ReducSystmSymme} writes in a hybrid form:
\begin{subequations}
\label{EQUATSS4.2SystmSymme_x_w_}
\begin{empheq}[left=\empheqlbrace]{align}
x' &= \displaystyle{\frac{rw - r}{(\alpha+1)w + x - \alpha}}, \label{EQUATSS4.2SystmSymme_x_w1}\\
y' &= \frac{w}{b}\cdotp\frac{\left(x + \frac{1}{2}w\right)}{\left( (\alpha+1) w + x - \alpha \right)^2} \label{EQUATSS4.2SystmSymme_x_w2} \cdotp
\end{empheq}
\end{subequations}
The computations to establish this new system are reported in the Appendix (see \ref{APPENSS___SymmeSystm(x,w)}). Note that for our purpose, it is not necessary to find the expression of $w'$ in terms of $x$ and $w$. It will be enough to use \eqref{EQUATSS4.2SystmSymme_x_w_} in order to determine the fixed points.\\
The first equation \eqref{EQUATSS4.2SystmSymme_x_w1} can be rewritten as:
\begin{align}
\label{EQUATSS4.2SystmFxdPtEstb1} 
x \left[ (\alpha+1) w + x - \alpha \right] = rw - r,
\end{align}
which corresponds to the equation of a conic $\mathcal{H}$ in the $(x,w)$ variables (namely, an hyperbola
). The second equation we find for the fixed points comes from a combination of the two equations of \eqref{EQUATSS4.2SystmSymme_x_w_}. On the one hand, if we look for fixed points of the form $(0,w)$, we obtain the equation
\begin{align}
rw-r = 0,
\end{align}
so that the only fixed point of this form is $(0,1)$. On the other hand, for $x=x' \neq 0$, we find, if $(x,w)$ is a fixed point:
\begin{align}
\frac{y'}{(x')^2} = \frac{w}{b} \cdotp \frac{\left(x + \frac{1}{2}w\right)}{\left( (\alpha+1) w + x - \alpha \right)^2} \cdotp \frac{\left( (\alpha+1) w + x - \alpha \right)^2}{(rw - r)^2} = \frac{y}{x^2},
\end{align}
where we have:
\begin{align*}
w = 1 - \sqrt{1 - 2by} \Longleftrightarrow y = \frac{2w-w^2}{2b},
\end{align*}
so that the equation $\frac{y'}{(x')^2}= \frac{y}{x^2}$ can be rearranged into
\begin{align}
\label{EQUATSS4.2SystmFxdPtEstb2}
2x^2w(x+\frac{1}{2}w) = (2w-w^2)(rw-r)^2.
\end{align}
After simplification, \eqref{EQUATSS4.2SystmFxdPtEstb1} and \eqref{EQUATSS4.2SystmFxdPtEstb2} provide the system for the fixed points in $\mathbb{R}_+^*\times\mathbb{R}_+$:
\begin{subequations}
\label{EQUATSS4.2SystmFxdPtFinal}
\begin{empheq}[left=\empheqlbrace]{align}
x \left[ (\alpha+1) w + x - \alpha \right] &= r(w - 1),\label{EQUATSS4.2SystmFxdPtFina1}\\
2x^2(x+\frac{1}{2}w) &= r^2(2-w)(w-1)^2. \label{EQUATSS4.2SystmFxdPtFina2}
\end{empheq}
\end{subequations}
The second equation \eqref{EQUATSS4.2SystmFxdPtFina2} corresponds to a cubic, and such an equation does \emph{not} depend on the parameter $\alpha$ (considering the equation $\frac{y'}{x'}=\frac{y}{x}$ would have also provided the equation of a cubic curve, but depending on the two parameters $r$ and $\alpha$). It is possible to determine precisely the number of solutions of \eqref{EQUATSS4.2SystmFxdPtFinal} with geometrical arguments. Indeed, on the one hand hand, the hyperbola $\mathcal{H}$ described by \eqref{EQUATSS4.2SystmFxdPtFina1} has the lines
\begin{align}
x = \frac{r}{\alpha+1} \hspace{2cm} \text{and} \hspace{2cm} (\alpha+1)w + x + \frac{r}{\alpha+1} - \alpha = 0
\end{align}
for asymptotes, and so there is one branch contained in the half-plane $\mathbb{R}_+\hspace{-0.5mm}\times\mathbb{R}$. Let us notice that, although the system \eqref{EQUATSS4.2SystmFxdPtFinal} makes sense for any value of $w$, we introduced this new variable with \eqref{EQUATSS4.4DefinVaria__w__}, implying in particular that $0 \leq w \leq 1$. But the other branch of $\mathcal{H}$ is always above $w=1$ for $x \geq 0$, therefore only the branch contained in the half-plane $\mathbb{R}_+\hspace{-0.5mm}\times\mathbb{R}$ might contain relevant solutions to our problem. Let us notice in particular that this branch can be parametrized by a function of the form $w = g_1(x)$ on $]r/(\alpha+1),+\infty[$, and that the derivative of this parametrization is always decreasing.\\
Turning now to the second equation \eqref{EQUATSS4.2SystmFxdPtFina2}, let us observe that the curve obtained is a cubic $\mathcal{C}$, with a double point, i.e., a crunodal cubic (see for instance \cite{Gibs998}, Chapter 15), and this double point lies at $(0,1)$. In the quadrant $\mathbb{R}_+\hspace{-0.5mm}\times]-\infty,1]$ is contained a single branch of this cubic, starting from the double point, and such that this branch is infinite and can be parametrized by a function of the form $w = g_2(x)$, always decreasing (see Appendix \ref{APPENSS___BrnchCubicDecrs}). Since the cubic $\mathcal{C}$ is crunodal, it is well-known that it has a single inflection point (cf. \cite{Gibs998}), and we can show that this point lies in the half-plane $\mathbb{R}_-\hspace{-1mm}\times\mathbb{R}$ (see Appendix \ref{APPENSS___InflcPointCubic} for more details). Since a tangent vector to the branch of $\mathcal{C}$ at $(0,1)$ is $(1,-1/r)$, and since the line:
\begin{align}
w = \mu x + C_r,
\end{align}
is the asymptote to the branch, where $\mu$ is the only real root of the polynomial equation $r^2 \mu^3 + \mu + 2 = 0$ and $C_r$ is a real constant 
depending only on $r$, we deduce that the derivative of the parametrization $w=g_2(x)$ is decreasing as $x$ increases.\\
But for all $0 < r < 1$ we have:
\begin{align}
\label{EQUATSS4.2CntrlRoot__mu__}
-\frac{1}{r} < \mu < -1,
\end{align}
because we have $r^2(-1/r)^3 + (-1/r) + 2 = 2(1-1/r) < 0$ and $r^2(-1)^3 + (-1) + 2 = 1-r^2 > 0$ for $0 < r < 1$, so that in particular the slope of the branch of the cubic is maximal at infinity. Concerning the hyperbola $\mathcal{H}$, its slope is always larger than
\begin{align}
- \frac{1}{\alpha+1},
\end{align}
so that the slope of the difference of the parametrizations $g_2-g_1$ is smaller than:
\begin{align}
\mu + \frac{1}{\alpha+1} = \frac{1}{\alpha+1} - \vert \mu \vert.
\end{align}
But $\alpha$ is positive and so according to \eqref{EQUATSS4.2CntrlRoot__mu__}, we deduce that the difference of the parametrizations $g_2-g_1$ is always strictly decreasing on its domain of definition $]r/(\alpha+1),+\infty[$. By the intermediate value theorem, we deduce that there exists a single solution $x_0$ such that $g_2(x_0)-g_1(x_0) = 0$, that is, there exists a single intersection point $(x_0,w_0)$ with $x_0 >0$ and $w_0 \leq 1$ between the hyperbola $\mathcal{H}$ and the cubic $\mathcal{C}$
defined by the equations \eqref{EQUATSS4.2SystmFxdPtFina1} and \eqref{EQUATSS4.2SystmFxdPtFina2}, and this point $(x_0,w_0)$ is such that
\begin{align}
x_0 > \frac{r}{\alpha+1}\cdotp
\end{align}
Now, such an intersection would be an acceptable fixed point of our symmetric iteration \eqref{EQUATSS4.2ReducSystmSymme} only if we have $w_0 \geq 0$. Let us note that the fixed points of the form $w_0 = 0$ have already been obtained in Section \ref{SSSCT4.1.3PremiProprDbIte}. In the present case, it is clear that a necessary condition to have a fixed point $(x_0,w_0)$ of the form $w_0 > 0$ is that the second branch of the hyperbola $\mathcal{H}$ intersects the first quadrant $x_0 > 0,w_0 > 0$, that is, its maximum is a positive number. The critical points of the hyperbola $\mathcal{H}$, seen as the graph of a function of $x$, are reached for
\begin{align}
x_{\pm} = \frac{r \pm \sqrt{r^2 + r(\alpha+1)}}{\alpha+1},
\end{align}
with only $\pm=+$ corresponding to a positive number. The maximum of $\mathcal{H}$ associated to $x_+$ is positive if and only if the second branch of the hyperbola intersects twice the first axis, that is, when the discriminant of the equation
\begin{align}
\label{EQUATSS4.2InterHyprb1stAx}
-x^2 + \alpha x - r = 0 
\end{align}
is positive, or again, when:
\begin{align}
\alpha^2 > 4r,
\end{align}
we recover the first condition of Zhou and Kadanoff, corresponding also to the existence of two fixed points for the iteration \eqref{EQUATSS4.2ReducSystmSymme} on the line $w=0$, i.e. on the line $\varphi_2=0$. However, it is \emph{not} a sufficient condition: the two fixed points, that we denote again by $p_-$ and $p_+$, of respective coordinates
\begin{align*}
p_- = \left(\varphi_1^-,0\right) \hspace{1cm} \text{and} \hspace{1cm} p_+ = \left(\varphi_1^+,0\right)
\end{align*}
where $\varphi_1^\pm$ are the two real roots of \eqref{EQUATSS4.2InterHyprb1stAx}, can exist on the invariant region $y = 0$, while the intersection $(x_0,w_0)$ between the hyperbola $\mathcal{H}$ and the cubic $\mathcal{C}$ is such that $w_0 < 0$. This intermediate case corresponds to the situation where there exist two equilibria in the invariant region $\varphi_2=0$, but none of them is stable in the plane $(\varphi_1,\varphi_2)$.\\
Let us denote by $p_\mathcal{C}$ the intersection point between the cubic $\mathcal{C}$ and the first axis, and let
\begin{align}
p_\mathcal{C} = \left(x_\mathcal{C},0\right)
\end{align}
be its coordinates. By our previous study of the branch of the cubic, or by direct computation, such a point exists, is unique, and is such that $x_\mathcal{C} > 0$. By the intermediate value theorem, a sufficient condition to have a non-trivial fixed point $(x_0,w_0)$ is that the cubic $\mathcal{C}$ intersects the first axis between the two points $p_-$ and $p_+$, the intersection points between this axis and the second branch of the hyperbola $\mathcal{H}$. This condition is actually an equivalence. Indeed, since the intersection point $p_\mathcal{C}$ is actually $(r^{2/3},0)$, this condition writes:
\begin{align}
\label{EQUATSS4.2EntrelacmtInter} 
\frac{\alpha - \sqrt{\alpha^2 - 4r}}{2} < r^{2/3} < \frac{\alpha + \sqrt{\alpha^2 - 4r}}{2} \cdotp
\end{align}
Note that by the assumption $\alpha^2 > 4r$, we have $\alpha^2 > 4r^{4/3}$, so that we have always:
\begin{align*}
2r^{2/3} - \alpha < 0 < \sqrt{\alpha^2 - 4r},
\end{align*}
in other words, we have always $x_\mathcal{C} < \varphi_1^+$, that is the intersection between the cubic and the first axis has always a smaller abscissa than the second intersection between this axis and the second branch of the hyperbola $\mathcal{H}$, if it exists. Therefore, if the condition \eqref{EQUATSS4.2EntrelacmtInter} does not hold, then the point $p_\mathcal{C}$ lies below the first intersection point $p_-$. Since $p_\mathcal{C}$ lies on the branch of the cubic we studied, we know that this branch lies in $w > 0$ if and only if $x < x_\mathcal{C}$, while the branch of the hyperbola lies below the first axis for this range of abscissas. In that case, we wouldn't have any fixed point $(x_0,w_0)$ with $w_0 > 0$.\\
Therefore, the condition \eqref{EQUATSS4.2EntrelacmtInter} is exactly equivalent to
\begin{align*}
0 < \alpha - 2r^{2/3} < \sqrt{\alpha^2 - 4r},
\end{align*}
that we can square and rearrange into
\begin{align}
\label{EQUATSS4.2CaracExistEquZK} 
r^{1/3} + r^{2/3} < \alpha.
\end{align}
The inequality \eqref{EQUATSS4.2CaracExistEquZK} characterizes then the existence of a non trivial fixed point to the symmetric case of the symmetric iteration \eqref{EQUATSS4.2ReducSystmSymme}. Let us denote by $p_{\text{unstab.}}^{ZK}$ such a fixed point, that we will call the Zhou-Kadanoff equilibrium (such a point was indeed observed and described already in \cite{ZhKa996}, see page 657, where the point is denoted by $(\alpha_0,k_0)$).\\
\newline
Let us now study the stability of the Zhou-Kadanoff equilibrium $p_{\text{unstab.}}^{ZK}$, adapting the argument of \cite{ZhKa996} to the setting of our variables. The goal is to show that one of the two eigenvalues of the Jacobian matrix of the iteration \eqref{EQUATSS4.2ReducSystmSymme}, computed at $p_{\text{unstab.}}^{ZK}$, is strictly larger than $1$, disproving immediatly the stability of such an equilibrium. Let us note $\text{Jac}(p_{\text{unstab.}}^{ZK})$ such a Jacobian matrix. To do so, the idea (already exposed in \cite{ZhKa996}) is to show that:
\begin{align}
\label{EQUATSS4.2InequStabiFxPZK} 
1 - \text{Tr}\left( \text{Jac}(p_{\text{unstab.}}^{ZK}) \right) + \text{det} \left( \text{Jac}(p_{\text{unstab.}}^{ZK}) \right) < 0,
\end{align}
which would imply, by continuity, that the characteristic polynomial:
\begin{align*}
\lambda^2 - \text{Tr}\left( \text{Jac}(p_{\text{unstab.}}^{ZK}) \right) \lambda + \text{det} \left( \text{Jac}(p_{\text{unstab.}}^{ZK}) \right)
\end{align*}
has a zero strictly larger than $1$.\\
If we denote by $F(x,y) = \left(f_1(x,y),f_2(x,y)\right)$ the symmetric iteration \eqref{EQUATSS4.2ReducSystmSymme}, \eqref{EQUATSS4.2InequStabiFxPZK} can be rewritten as:
\begin{align}
\left( 1 - \partial_x f_1 \right) \left( 1 - \partial_y f_2 \right) - \partial_y f_1 \partial_x f_2 < 0,
\end{align}
where all the partial derivatives are evaluated at $p_{\text{unstab.}}^{ZK}$.\\
In the present case, the partial derivatives write:
\begin{align*}
\partial_x f_1 &= \frac{r(1-w)}{\left( (\alpha+1) w + x - \alpha \right)^2} ,\\
\partial_y f_1 &= \frac{rb (x+1)}{(1-w) \left( (\alpha+1) w + x - \alpha \right)^2} ,\\
\partial_x f_2 &= \frac{w \left( \alpha(w-1) - x \right)}{b \left( (\alpha+1) w + x - \alpha \right)^3} ,\\
\partial_y f_2 &= \frac{(x+w) \left( (\alpha+1) w + x - \alpha \right) - 2(\alpha+1) w \displaystyle{\left( x + \frac{w}{2} \right)}}{(1-w) \left( (\alpha+1) w + x - \alpha \right)^3} \cdotp
\end{align*}
Let us note on the one hand that $\partial_y f_1$ is unconditionally positive (because we have $0 \leq w \leq 1$), and on the other hand $\partial_x f_2 > 0$ while restricted on the hyperbola $\mathcal{H}$, since the equation \eqref{EQUATSS4.2SystmFxdPtFina1} provides:
\begin{align}
\label{EQUATSS4.2DenomNegatHyprb}
(\alpha+1)w + x - \alpha = \frac{r(w-1)}{x} < 0,
\end{align}
and the line $\alpha(w-1) - x = 0$ is always above the area $x > 0, 0 < w < 1$ of our study, so that $\alpha(w-1)-x < 0$ evaluated on $p_{\text{unstab.}}^{ZK}$.\\
Let us now turn to the term $1 - \partial_x f_1$. On the hyperbola $\mathcal{H}$, we find:
\begin{align}
1 - \partial_x f_1 = \frac{(\alpha+1)w + x - \alpha - x}{\left((\alpha+1) w + x - \alpha \right)} = (\alpha+1) \frac{w - \frac{\alpha}{\alpha+1}}{\left((\alpha+1) w + x - \alpha \right)} > 0,
\end{align}
using to conclude that, on the hyperbola $\mathcal{H}$, on the one hand we have once again that the denominator is negative, and on the other hand $w < \frac{\alpha}{\alpha+1}$ (because the maximum of the second branch of the hyperbola is always below such an ordinate).\\
Finally, turning to the term $1 - \partial_y f_2$, we can write:
\begin{align*}
1& - \partial_y f_2 = \frac{(1-w) \left( (\alpha+1) w + x - \alpha \right)^3 - (x+w) \left( (\alpha+1) w + x - \alpha \right) + 2 (\alpha+1) w \displaystyle{\left(x + \frac{w}{2}\right)}}{(1-w) \left( (\alpha+1) w + x - \alpha \right)^3} \\
&= \frac{(1-w) r^2 (w-1)^2 \left( (\alpha+1) w + x - \alpha \right) - x^2 (x+w) \left( (\alpha+1) w + x - \alpha \right) + (\alpha+1) w r^2 (2-w) (w-1)^2}{x^2 (1-w) \left( (\alpha+1) w + x - \alpha \right)^3},
\end{align*}
using, first that the fixed point $p_{\text{unstab.}}^{ZK}$ is on the hyperbola $\mathcal{H}$ so that we rewrite the first term of the numerator according to \eqref{EQUATSS4.2SystmFxdPtFina1}, and second that $p_{\text{unstab.}}^{ZK}$ is also on the cubic $\mathcal{C}$ so that we used \eqref{EQUATSS4.2SystmFxdPtFina2} to replace the third term. Now, the first term is clearly negative on the hyperbola $\mathcal{H}$, while the second and third terms provide together:
\begin{align*}
(\alpha+1) &r^2 w (2-w) (w-1)^2 - x^2 (x+w) \left( (\alpha+1) w + x - \alpha \right) \\
&= (\alpha + 1) r^2 w (2-w) (w-1)^2 + x^3 \left( (\alpha+1) w + x - \alpha \right) - x^2 (2x+w) \left( (\alpha+1) w + x - \alpha \right) \\
&= (\alpha + 1) r^2 w (2-w) (w-1)^2 + x^3 \left( (\alpha+1) w + x - \alpha \right) -r^2(2-w)(w-1)^2
\end{align*}
using that $p_{\text{unstab.}}^{ZK}$ is on the cubic $\mathcal{C}$, so that in the end:
\begin{align*}
(\alpha+1) &r^2 w (2-w) (w-1)^2 - x^2 (x+w) \left( (\alpha+1) w + x - \alpha \right) \\
&= r^2 (2-w)(w-1)^2 \left( (\alpha + 1)w - 1 \right) + x^3 \left( (\alpha+1) w + x - \alpha \right),
\end{align*}
and we see that both terms are negative on the hyperbola, because on the one hand $w < \frac{\alpha}{\alpha + 1} \leq \frac{1}{\alpha+1}$ concerning the first term, and as for the second, we conclude using \eqref{EQUATSS4.2DenomNegatHyprb}.\\
In summary, we have:
\begin{align*}
\underbrace{\left( 1 - \partial_x f_1 \right)}_{>0} \underbrace{\left( 1 - \partial_y f_2 \right)}_{<0} - \underbrace{\partial_y f_1 \partial_x f_2}_{>0} < 0,
\end{align*}
that is, \eqref{EQUATSS4.2InequStabiFxPZK} holds, so the Jacobian matrix of the symmetric iteration \eqref{EQUATSS4.2ReducSystmSymme} at the Zhou-Kadanoff equilibrium $p_{\text{unstab.}}^{ZK}$ has one real eigenvalue larger than $1$, so that this equilibrium is not stable. This completes the proof of Theorem \ref{THEORSS4.2PointFixe_ZhKad}.
\end{proof}

\subsection{Another unstable fixed point}

To conclude the study of the symmetric reduced system, let us notice that if we consider the symmetric iteration \eqref{EQUATSS4.2ReducSystmSymme} on the domain $\varphi_1 \geq 0$, $\varphi_2 \geq 0$ (that is, if we include also the line $\varphi_1 = 0$), there is another fixed point. Let us provide a brief study of this equilibrium.

\begin{propo}[Existence of a fixed point on the line $\varphi_1=0$]
Let us consider the symmetric case $a=b$ of the two-collision mapping \eqref{EQUATSS4.1Itera_2dim2Coll}, which reduces to the single mapping \eqref{EQUATSS4.2ReducSystmSymme}.\\
There exists a last fixed point $p_\text{up} = \left(0,1/2b\right)$ of the mapping \eqref{EQUATSS4.2ReducSystmSymme}, located on the line $\varphi_1=0$, and this fixed point is unstable in the domain $\varphi_1 \geq 0$, $\varphi_2 \leq 1/2b$.
\end{propo}

\begin{proof}
The fact that $p_\text{up} = \left(0,1/2b\right)$ is a fixed point is direct from plugging the coordinates into \eqref{EQUATSS4.2ReducSystmSymme}.\\
Concerning the stability, we simply consider a point of the form $\left(0,1/2b-\varepsilon\right)$, with $\varepsilon$ small enough. Using the $(x,w)$ coordinates of the proof of Theorem \ref{THEORSS4.2PointFixe_ZhKad}, we find:
\begin{align}
\partial_y f_2(0,1) &= \frac{(1-\varepsilon)\left((\alpha+1)(1-\varepsilon)-\alpha\right) - 2 (\alpha+1) (1-\varepsilon) \left(0 + \frac{(1-\varepsilon)}{2}\right)}{\varepsilon\left((\alpha+1)(1-\varepsilon) - \alpha\right)} \\
&= \frac{(1-\varepsilon)(-\alpha)}{\varepsilon\left( 1 -\alpha \varepsilon - \varepsilon \right)} < 0
\end{align}
(since the point $(0,1)$ in the $(x,w)$ coordinates is exactly the fixed point $\left(0,1/2b\right)$ in the usual variables $(\varphi_1,\varphi_2)$). This last computation shows that any point below the fixed point $p_\text{up}$, close enough, will be repelled downwards by this fixed point, hence the instability.
\end{proof}

\section{A formal limit of the two-collision mapping \eqref{EQUATSS4.1Itera_2dim2Coll}: the low energy limit}
\label{SSECTIO4.3FormlLimtLowNRG}

In this section we present a dynamical system formally obtained from the two-collision mapping \eqref{EQUATSS4.1Itera_2dim1Col1}, \eqref{EQUATSS4.1Itera_2dim1Col2} in the so-called regime of the \emph{low energy limit}. We define the low energy limit as follows: we fix the ratio $T=b/a$, where $a$ and $b$ are by definition the squares of the norms of the respective tangential velocities of the two external particles \circled{1} and \circled{2}, measured with respect to the central particle \circled{0}. Since the normal components of the relative velocities are completely dissipated at the time of collapse, the data of $a$ and $b$ encode the remaining energy in the system. Keeping the ratio $T = b/a$ fixed, and sending $b$ to zero, we obtain the regime of \emph{low energy}.\\
We will see that for such a system, we will be able to prove the conjectures stated in Section \ref{SSSCT4.1.4NumerConjectuDCM} for the two-collision mapping.

\subsection{Equations of the reduced system in the low energy limit}

In this section, let us establish the equations of the two-collision mapping \eqref{EQUATSS4.1Itera_2dim2Coll} in the \emph{low energy limit}. Up to exchange the label of the two external particles, we can always assume that $b \leq a$, that is, $T \leq 1$. The starting point is the first iteration \eqref{EQUATSS4.1Itera_2dim1Coll} defining the dynamical system. Assuming that the orbits remain bounded (that is, that the second component $\varphi_2$ remains bounded) and using the fact that $\alpha - \varphi_1$ is bounded from below, we send $b$ to infinity, so that we obtain for the first coordinate:
\begin{align*}
\varphi_1' \simeq \frac{r}{\alpha + \frac{1}{T}(-b\varphi_2) - \varphi_1} \simeq \frac{r}{\alpha - \varphi_1},
\end{align*}
since $T = b/a$ so that $b/T = a$, and $b \rightarrow 0$ with $T > 0$ fixed implies also $a \rightarrow 0$. In the present case, the expression of $\varphi_1'$ simplifies, and we do not obtain a finer approximation than when assuming $\varphi_2=0$. For the second coordinate we obtain a new expression, which is more complex than $\varphi_2'=0$, but simpler than \eqref{EQUATSS4.1Itera_2dim1Col2}:
\begin{align*}
\varphi_2' \simeq \frac{b \varphi_2}{b} \cdot \frac{\left( \varphi_1 + \frac{1}{2T}(b \varphi_2) \right)}{ \left( \alpha - \frac{b}{T} \varphi_2 - \varphi_1 \right)^2} \simeq \frac{\varphi_1 \varphi_2}{(\alpha - \varphi_1)^2} \cdotp
\end{align*}
We can now insert these two last expressions in the equations \eqref{EQUATSS4.1IteraIntermPhi1} and \eqref{EQUATSS4.1IteraIntermPhi2} describing the second iteration, and we obtain
\begin{align*}
\varphi_1'' &\simeq \frac{r \sqrt{ 1 - 2a \frac{\varphi_1\varphi_2}{(\alpha-\varphi_1)^2} }}{(\alpha+T) \sqrt{1 - 2a\frac{\varphi_1\varphi_2}{(\alpha-\varphi_1)^2}} - \frac{r}{\alpha-\varphi_1} - T} \\
&\simeq \frac{r}{\alpha - a T \frac{\varphi_1\varphi_2}{(\alpha-\varphi_1)^2} - \frac{r}{\alpha-\varphi_1}} \simeq \frac{r}{\alpha - \frac{r}{\alpha-\varphi_1}} = \frac{r(\alpha-\varphi_1)}{\alpha(\alpha-\varphi_1) - r} \cdotp
\end{align*}
Therefore, for the first coordinate of the two-collision mapping, the low energy limit corresponds to the reduced dynamics on the region $\varphi_2 = 0$. As a consequence, the variables $\varphi_1$ and $\varphi_2$ are decoupled, in the sense that the evolution of $\varphi_1$ is now independent from $\varphi_2$ (the converse will not be true, as we will see).\\
In the case of the second coordinate, we obtain:
\begin{align*}
\varphi_2'' \simeq \frac{\varphi_1\varphi_2}{(\alpha-\varphi_1)^2} \cdotp \frac{\left( \frac{r}{\alpha-\varphi_1} + \frac{T}{2}\frac{b}{T}\varphi_2' \right)}{\left(\alpha - \frac{r}{\alpha-\varphi_1}\right)^2} \simeq \frac{r \varphi_1\varphi_2}{(\alpha-\varphi_1)\left( \alpha(\alpha-\varphi_1)-r \right)^2} \cdotp
\end{align*}
In the end, the two-collision mapping $\left(\varphi_1,\varphi_2\right) \mapsto \left(\varphi_1'',\varphi_2''\right)$ writes in the low energy limit:
\noindent
\begin{align}
\label{EQUATSS4.3IteraBasseEnerg}
\left\{
\begin{array}{rl}
\varphi_1'' &\simeq \displaystyle{\frac{r(\alpha-\varphi_1)}{\alpha(\alpha-\varphi_1) - r}}, \\
\varphi_2'' &\simeq \displaystyle{\frac{r \varphi_1\varphi_2}{(\alpha-\varphi_1)\left( \alpha(\alpha-\varphi_1)-r \right)^2}} \cdotp
\end{array}
\right.
\end{align}
\noindent
Let us note that the equations \eqref{EQUATSS4.3IteraBasseEnerg} of this new dynamical system can be simplified, reformulating the system with a new set of variables. Indeed, defining $X = \varphi_1/\alpha$, we obtain:
\begin{align*}
\varphi_1'' = \frac{r(1-\varphi_1/\alpha)}{\alpha(1-\varphi_1/\alpha)-r/\alpha} = \frac{r(1-X)}{\alpha(1-X) - r/\alpha},
\end{align*}
or again, writing $X' = \varphi_1''/\alpha$:
\begin{align*}
X' = \varphi_1''/\alpha = \frac{1}{\alpha} \frac{r(1-X)}{\alpha((1-X)-r/\alpha^2)},
\end{align*}
so that, writing also $R = r/\alpha^2$, we find:
\begin{align}
X' = \frac{R(1-X)}{1-X-R} \cdotp
\end{align}
The same simplification can be carried on for the second coordinate. Restarting from the new variables we just introduced, we get:
\begin{align*}
\varphi_2'' &= \frac{\alpha^3(r/\alpha^2)(\varphi_1/\alpha)\varphi_2}{\alpha(1-\varphi_1/\alpha)(\alpha^2(1-\varphi_1/\alpha)-r)^2} \\
&= \frac{\alpha^2 R X \varphi_2}{(1-X)\alpha^4(1-X-R)^2},
\end{align*}
so that in the end, writing $Y = \varphi_2$, we obtain a synthetic expression for the two-collision mapping.\\
Let us summarize these formal computations in the following definition.

\begin{defin}[Non-dimensional, low energy limit of the two-collision mapping \eqref{EQUATSS4.1Itera_2dim2Coll}]
Let $R,\alpha \in\ ]0,1[$ be two strictly positive numbers, strictly smaller than $1$.\\
We define the \emph{non-dimensional, low energy limit of the two-collision mapping} as the mapping $\left(\mathbb{R}_+\right)^2 \rightarrow \mathbb{R}^2$, $\left(X,Y\right) \mapsto \left(X',Y'\right)$, where:
\begin{align}
\label{EQUATSS4.3IteraBEnrgAdimn}
\left\{
\begin{array}{rl}
X' &= \displaystyle{\frac{R(1-X)}{1-X-R}}, \\
Y' &= \displaystyle{\frac{(R/\alpha^2) XY}{(1-X)(1-X-R)^2}} \cdotp
\end{array}
\right.
\end{align}
\end{defin}

\subsection{Long time behaviour of the orbits of the reduced system in the low energy limit}

The dynamical system defined with the mapping \eqref{EQUATSS4.3IteraBEnrgAdimn} turns out to be easy to understand: as we noticed already the first variable is independent from the second, and presents two fixed points $X^\pm$, solutions of the polynomial equation $X(1-X-R) = R(1-X)$, that is $X^2-X+R=0$, where only the smallest, $X^-$, is stable. On the other hand, the dynamics of the second variable is deduced from the value of
\begin{align*}
\frac{R X^-}{\alpha^2 (1-X^-)(1-X^-R)^2} \cdotp
\end{align*}
We obtain the following result.

\begin{theor}[Dynamics of the non-dimensional, low energy limit of the two-collision mapping]
\label{THEORSS4.3DynamBEnrgAdimn}
Let $R,\alpha \in\ ]0,1[$ be two strictly positive numbers, strictly smaller than $1$. Let us also assume that $R < 1/4$, so that the non-dimensional, low energy limit of the two-collision mapping has two fixed points of the form $\left(X^\pm,0\right)$, with:
\begin{align}
X^\pm = \frac{1 \pm \sqrt{1-4R}}{2},
\end{align}
where $0 < R < X^- < 1/2 < X^+ < 1-R < 1$. $X^-$ is the unique stable fixed point of the stable manifold $Y=0$. Let us denote by $C(R,\alpha)$ the positive quantity:
\begin{align}
C(R,\alpha) = \frac{\left( R/\alpha^2 \right) X^-}{(1-X^-)(1-X^--R)^2}\cdotp
\end{align}
Let us finally consider the orbit $(X_n,Y_n)_n$ of the dynamical system defined by the non-dimensional, low energy limit of the two-collision mapping, starting from the initial datum $(X_0,Y_0) \in [0,1[\times [0,+\infty[$.
Then, if $C(R,\alpha) < 1$, the orbit $(X_n,Y_n)_n$ will converge to $(X^-,0)$, exponentially fast.\\
Finally, the two vertical lines $X=X^-$ and $X=X^+$ are stable manifolds.
\end{theor}

\begin{remar}
In the case when $C(R,\alpha) > 1$, the second component $Y_n$ of the orbit will eventually tend to $+\infty$, exponentially fast. Let us observe that in this case, the formal approximation we made to obtain the low energy limit breaks down (the boundedness of the orbits does not hold).\\
On the other hand, the condition $C(R,\alpha) < 1$ writes explicitly:
\begin{align*}
\frac{R}{\alpha^2} X^- < (1-X^-)(1-X^--R)^2,
\end{align*}
where $(X^-)^2 - X^- + R = 0$. Therefore, it is not completely clear a priori that this condition can be fulfilled for some $R$ and $\alpha$. Nevertheless, this condition holds true if, for instance:
\begin{align*}
R^2 < \frac{\alpha^2}{64}\cdotp
\end{align*}
Hence, for a restitution coefficient small enough, and a limiting angle between the particles obtuse enough, the Zhou-Kadanoff regime is certain for the non-dimensional, low energy limit of the two-collision mapping.
\end{remar}

\begin{remar}
Let us recall that the mapping \eqref{EQUATSS4.3IteraBEnrgAdimn} corresponds to the low energy limit $a,b \rightarrow 0$. Theorem \ref{THEORSS4.3DynamBEnrgAdimn} proves that the Zhou-Kadanoff regime happens eventually for \emph{all} the orbits of the mapping \eqref{EQUATSS4.3IteraBEnrgAdimn} starting from $[0,1[\times[0,+\infty[$ (if $C(R,\alpha)$ is strictly smaller than $1$). On the other hand, the vertical line $X=X^+$ is a separatrix, which partitions the first quadrant: for $X < X^+$, the Zhou-Kadanoff regime is certain, for $X > X^+$, the dynamics is not globally well-posed, in the sense that the orbit will leave the first quadrant.\\
Let us note finally that the results obtained for the low energy limit of the two-collison mapping agree with what we observed numerically and reported in Section \ref{SSSCT4.1.4NumerConjectuDCM} concerning the two-collision mapping \eqref{EQUATSS4.1Itera_2dim2Coll}. Indeed, if we consider that the low energy limit induces a non-linear dilatation of the second variable $\varphi_2$ and sends in particular the unstable Zhou-Kadanoff equilibrium to infinity, we can identify the stable manifolds $\mathcal{S}$ and $\mathcal{V}$ described in \ref{SSSCT4.1.4NumerConjectuDCM} with the stable lines $X=X^+$ and $X=X^-$. Indeed, in the case of the low energy limit, the line $X=X^+$ plays clearly the role of a separatrix, where the orbits are globally defined only when they start from $(X_0,Y_0)$ with $X_0 \leq X^+$. The counterpart of this approach is that a part of the behaviour of the orbits of the two-collision mapping is not captured in the low energy limit, since the unstable Zhou-Kadanoff equilibrium and its neighbourhood are sent to infinity. However, in the low energy limit we recover also the qualitative behaviour of the orbits of the original two-collision mapping that we observed numerically: in the case of the low energy limit, orbits starting close to the stable fixed point $(X^-,0)$ converge towards this point, vertically, while for $X_0$ large enough (but below $X^+$), the orbits start first to climb for a finite number of iterations, and end up finally to dive down to the fixed point, at the moment when $X_n$ becomes small enough (because it converges towards $X^-$) so that
\begin{align*}
\frac{(R/\alpha^2)X_n}{(1-X_n)(1-X_n-R)^2}
\end{align*}
becomes smaller than $1$, inducing the eventual exponential decay of the second component $Y_n$ of the orbits.\\
These common features exhibit the relevance of the low energy limit, for which we recover rigorously the structure of the phase space of the mapping \eqref{EQUATSS4.1Itera_2dim2Coll}, divided into two parts by the separatrix, as well as the behaviour of the different orbits.
\end{remar}

\section{Conclusion and perspective}

Together with the companion paper \cite{DoVeAr1}, we investigated the collapse phenomenon of three inelastic particles in dimension $d \geq 2$. In the present article, we studied in particular the only collapsing configuration that we expect to be generic, that is, associated to a set of initial configurations of positive measure: the case of the nearly-linear inelastic collapse, when the relative velocities of the particles are not vanishing at the time of the collapse.\\
To study such a system, we relied on the rigorous study of the asymptotics of the different variables (obtained in \cite{DoVeAr1}) to describe inelastic collapse. In particular, we used these asymptotics to identify the leading order terms of the full dynamical system describing a nearly-inelastic collapse. In turn, the evolution laws of the leading order terms provided a reduced two-dimensional system, the two-collision mapping that we introduced in this article. We studied this dynamical system, identifying formally the Zhou-Kadanoff regime, in which the two-collision mapping is completely understood. We attached then ourselves to understand the orbits of the two-collision mapping, and in particular to characterize completely the initial configurations that satisfy the Zhou-Kadanoff regime at the time of collapse.\\
We were able to prove that the Zhou-Kadanoff regime is stable and locally attracting in a non-trivial region of the phase space of the two-collision mapping. In order to overcome the limitation of the local nature of such a result, we performed numerical simulations of the two-collision mapping, which allowed to state a conjecture concerning the Zhou-Kadanoff regime. The numerical simulations suggest that there exists a separatrix that divides the phase space into two parts: one of them contains the initial configurations of orbits that are globally well-defined and eventually satisfy the Zhou-Kadanoff regime, and the other part contains the initial configurations of orbits that are not globally well-defined. We were able to prove this conjecture in the low energy limit, that is, we proved the conjecture for the formal limit of the two-collision mapping we obtain when we let the energy of the system going to zero.\\
We achieved a first step in understanding the inelastic collapse of three inelastic particles. We believe that, except for a set of initial configurations of zero measure, all the initial configurations leading to an inelastic collapse are such that, at the limiting time of the collapse, the relative velocities of the particles that are in contact are non zero. If such a property holds true, it would be possible to continue the dynamics of the systems of inelastic hard spheres, even in the case when an inelastic collapse takes place. To complete the proof of this continuation property, it would be necessary to prove the following results.
\begin{itemize}
\item We proved already in the companion paper \cite{DoVeAr1} that the order of collisions of a collapsing system of three inelastic particles eventually becomes (up to relabel the particles), either the infinite repetition of the collision sequence \circled{0}-\circled{1}, \circled{0}-\circled{2} (the nearly-linear inelastic collapse), or \circled{0}-\circled{1}, \circled{0}-\circled{2}, \circled{1}-\circled{2} (the triangular collapse). The study of the velocity matrix associated to the triangular collapse presented in \cite{DoVeAr1} strongly suggests that the triangular collapse is unstable, that is, can only be achieved (if it ever can be achieved) from a negligible set of initial configurations. In the future, we would like to prove that the triangular collapse is indeed unstable. Studying the triangular collapse can be hopefully achieved by extending the methods developed in the present paper.
\item The study of the two-collision mapping suggests that the Zhou-Kadanoff regime is the only stable outcome for a collapsing system of three particles (at least, under the assumption that none of the relative velocities vanishes). It would be interesting to prove such a result, already for the orbits of the two-collision mapping, and then for the full dynamical system describing an inelastic collapse. In particular, it seems reasonable to look for the separatrix we observed numerically, that seems to characterize the Zhou-Kadanoff regime. We encountered severe difficulties due to the complicated expression of the two-collision mapping. A possible approach to determine the separatrix may be to rely on computer assisted proofs. Another approach would be to apply a perturbation method, starting from the dynamical system obtained in the low energy limit, in order to recover the two-collision mapping as a perturbation of the low energy limit, at least for the particular case when the relative velocities are small at the time of collapse.
\item Finally, the two-collision mapping is the two-dimensional reduction of the full dynamical system provided that the norms of the relative velocities of the external particles do not vanish at the time of inelastic collapse. However, a priori the relative velocities can also vanish at the time of collapse. The reduced two-dimensional systems obtained in the case when one or the two relative velocities vanish have also to be studied.
\end{itemize}

\begin{appendices}

\section{Details of the computations concerning the two-collision mapping of the linear collapse \eqref{EQUATSS4.1Itera_2dim2Coll}}
\label{APPENSC___DetaiCalcuDbIte}

\subsection{Details of the computation \eqref{EQUATSS4.1DbIteTerme1-2aP}}
\label{APPENSS___CalcuDbIteExpre}

The details to establish the equation \eqref{EQUATSS4.1DbIteTerme1-2aP} page \pageref{EQUATSS4.1DbIteTerme1-2aP} are as follows:
\begin{align*}
1-2a\varphi_2' =& \left( \left(\alpha + \frac{1}{T}\right)\sqrt{1-2b\varphi_2} - \varphi_1 - \frac{1}{T} \right)^{-2} \Bigg[ \left(\alpha+\frac{1}{T}\right)^2(1-\zeta) + \varphi_1^2 + \frac{1}{T^2} - 2\left(\alpha+\frac{1}{T}\right) \varphi_1 \sqrt{1-\zeta} \nonumber\\
&\hspace{30mm} - 2 \left(\alpha+\frac{1}{T}\right) \frac{1}{T}\sqrt{1-\zeta} + 2 \frac{\varphi_1}{T} - 2 \frac{\varphi_1}{T}\left[1 - \sqrt{1 - \zeta}\right] - \frac{1}{T^2}\left[1 - \sqrt{1-\zeta}\right]^2 \Bigg] \nonumber\\
=& \left( \left(\alpha + \frac{1}{T}\right)\sqrt{1-2b\varphi_2} - \varphi_1 - \frac{1}{T} \right)^{-2} \Bigg[ \left(\alpha^2 + 2\frac{\alpha}{T}\right)(1-\zeta) + \frac{1}{T^2}(1-\zeta) + \varphi_1^2 + \frac{1}{T^2} \nonumber\\
&\hspace{25mm} - 2\left(\alpha+\frac{1}{T}\right) \varphi_1 \sqrt{1-\zeta} - 2 \left(\alpha+\frac{1}{T}\right) \frac{1}{T}\sqrt{1-\zeta} + 2 \frac{\varphi_1}{T} - 2 \frac{\varphi_1}{T}\left[1 - \sqrt{1 - \zeta}\right] \nonumber\\
&\hspace{95mm} - \frac{1}{T^2}\left(1 - 2\sqrt{1-\zeta} + (1-\zeta)\right) \Bigg] \nonumber\\
=& \left( \left(\alpha + \frac{1}{T}\right)\sqrt{1-2b\varphi_2} - \varphi_1 - \frac{1}{T} \right)^{-2} \Bigg[ \left(\alpha^2 + 2\frac{\alpha}{T}\right)(1-\zeta) + \varphi_1^2 - 2\left(\alpha+\frac{1}{T}\right) \varphi_1 \sqrt{1-\zeta} \nonumber\\
&\hspace{50mm} - 2 \left(\alpha+\frac{1}{T}\right) \frac{1}{T}\sqrt{1-\zeta} + 2 \frac{\varphi_1}{T} - 2 \frac{\varphi_1}{T}\left[1 - \sqrt{1 - \zeta}\right] + 2\frac{\sqrt{1-\zeta}}{T^2} \Bigg] \nonumber\\
=& \left( \left(\alpha + \frac{1}{T}\right)\sqrt{1-2b\varphi_2} - \varphi_1 - \frac{1}{T} \right)^{-2} \Bigg[ \left(\alpha^2 + 2\frac{\alpha}{T}\right)(1-\zeta) + \varphi_1^2 - 2\left(\alpha+\frac{1}{T}\right) \varphi_1 \sqrt{1-\zeta} \nonumber\\
&\hspace{60mm} - 2 \frac{\alpha}{T}\sqrt{1-\zeta} + 2 \frac{\varphi_1}{T} - 2 \frac{\varphi_1}{T}\left[1 - \sqrt{1 - \zeta}\right] \Bigg] \nonumber\\
=& \left( \left(\alpha + \frac{1}{T}\right)\sqrt{1-2b\varphi_2} - \varphi_1 - \frac{1}{T} \right)^{-2} \Bigg[ \alpha\left(\alpha + \frac{2}{T}\right)(1-\zeta) + \varphi_1^2 - 2\left(\alpha+\frac{1}{T}\right) \varphi_1 \sqrt{1-\zeta} \nonumber\\
&\hspace{60mm} - 2 \frac{\alpha}{T}\sqrt{1-\zeta} + 2 \frac{\varphi_1}{T}\sqrt{1 - \zeta} \Bigg] \nonumber\\
=& \left( \left(\alpha + \frac{1}{T}\right)\sqrt{1-2b\varphi_2} - \varphi_1 - \frac{1}{T} \right)^{-2} \Bigg[ \alpha \left(\alpha+\frac{2}{T}\right)(1-\zeta) - 2\alpha\left( \varphi_1 + \frac{1}{T} \right)\sqrt{1-\zeta} + \varphi_1^2 \Bigg].
\end{align*}

\subsection{Details of the computations of the Jacobian matrix of the mapping \eqref{EQUATSS4.1Itera1CollCoord}}
\label{APPENSS___CalcuDbIteJacob}

Only the computations of the partial derivatives $\partial_xf_1$, $\partial_xf_2$ and $\partial_yf_2$ of the mapping \eqref{EQUATSS4.1Itera1CollCoord} page \pageref{EQUATSS4.1Itera1CollCoord} are not direct.\\
Concerning the second derivative of the first coordinate $f_1$ we have:
\begin{align*}
\partial_y f_1(x,y) &= \frac{r \frac{(-2b)}{2\sqrt{1-2by}} \left[ \left( \alpha + \frac{1}{T} \right) \sqrt{1-2by} - x - \frac{1}{T} \right] - r \sqrt{1-2by} \left( \alpha + \frac{1}{T} \right) \frac{(-2b)}{2\sqrt{1-2by}}}{\left( \left( \alpha + \frac{1}{T} \right) \sqrt{1-2by} - x - \frac{1}{T} \right)^2} \\
&= \frac{- \frac{rb}{\sqrt{1-2by}}\left( \alpha + \frac{1}{T} \right) \sqrt{1-2by} + \left( x + \frac{1}{T} \right) \frac{rb}{\sqrt{1-2by}} + \frac{rb}{\sqrt{1-2by}} \sqrt{1-2by} \left( \alpha + \frac{1}{T} \right) }{\left( \left( \alpha + \frac{1}{T} \right) \sqrt{1-2by} - x - \frac{1}{T} \right)^2} \\
&= \frac{rb \left( x + \frac{1}{T} \right)}{ \sqrt{1-2by} \left( \left( \alpha + \frac{1}{T} \right) \sqrt{1-2by} - x - \frac{1}{T} \right)^2 } \cdotp
\end{align*}
Concerning the second coordinate $f^1_2$, we find first
\begin{align*}
\partial_x f^1_2(x,y) &= \frac{ \left(1-\sqrt{1-2by}\right) \left( \left( \alpha + \frac{1}{T} \right) \sqrt{1-2by} - x - \frac{1}{T} \right)^2}{ \left( \left(\alpha + \frac{1}{T}\right) \sqrt{1-2by} - x - \frac{1}{T} \right)^3 } \\
& \hspace{3mm} - \frac{ \left( 1 - \sqrt{1-2by} \right) \left( x + \frac{1}{2T} \left( 1 - \sqrt{1-2by} \right) \right) \cdot (-2) \left( \left( \alpha + \frac{1}{T} \right)\sqrt{1-2by} - x - \frac{1}{T} \right) }{ \left( \left( \alpha + \frac{1}{T} \right) \sqrt{1-2by} - x - \frac{1}{T} \right)^4} \\
&= \frac{  \left(1-\sqrt{1-2by}\right) \left[ \left( \alpha + \frac{1}{T} \right) \sqrt{1-2by} - x - \frac{1}{T} + 2 \left( x + \frac{1}{2T} \left( 1 - \sqrt{1-2by} \right) \right) \right] }{ \left( \left(\alpha + \frac{1}{T}\right) \sqrt{1-2by} - x - \frac{1}{T} \right)^3 } \\
&= \frac{\left( 1 - \sqrt{1-2by} \right) \left[ \alpha \sqrt{1-2by} + x \right]}{ \left( \left( \alpha + \frac{1}{T} \right) \sqrt{1-2by} - x - \frac{1}{T} \right)^3 } \cdotp
\end{align*}
For the second derivative of the last coordinate, we start with this intermediate computation:
\begin{align*}
&\partial_y \left[ \left(1-\sqrt{1-2by}\right) \left( x + \frac{1}{2T}\left(1-\sqrt{1-2by}\right) \right) \right] \\
&= - \frac{(-2b)}{2\sqrt{1-2by}} \left( x + \frac{1}{2T}\left(1-\sqrt{1-2by}\right)\right) + \left(1-\sqrt{1-2by}\right) \left( \frac{b}{2T\sqrt{1-2by}} \right) \\
&= \frac{bx}{\sqrt{1-2by}} + \frac{b\left(1-\sqrt{1-2by}\right)}{2T\sqrt{1-2by}} + \frac{b\left(1-\sqrt{1-2by}\right)}{2T\sqrt{1-2by}}\\
&= b\frac{\left(Tx + 1 - \sqrt{1-2by}\right)}{T\sqrt{1-2by}} \cdotp
\end{align*}
We find therefore
\begin{align*}
\partial_y f^1_2(x,y) &= \frac{1}{b^2 \left(\left(\alpha + \frac{1}{T}\right)\sqrt{1-2by} - x - \frac{1}{T} \right)^4} \Bigg[ b \frac{(Tx+1-\sqrt{1-2by})}{T\sqrt{1-2by}} b \left( \left(\alpha + \frac{1}{T}\right)\sqrt{1-2by} - x - \frac{1}{T} \right)^2 \\
&\hspace{0mm} - \left( 1 - \sqrt{1-2y} \right) \left( x + \frac{1}{2T}\left( 1 - \sqrt{1-2by} \right)\right) 2b \left( \left(\alpha + \frac{1}{T}\right)\sqrt{1-2by} - x - \frac{1}{T} \right) (\alpha+1/T)\frac{(-b)}{\sqrt{1-2by}} \Bigg] \\
&= \frac{1}{\left(\left(\alpha + \frac{1}{T}\right)\sqrt{1-2by} - x - \frac{1}{T} \right)^3} \Bigg[ \frac{(Tx+1-\sqrt{1-2by})}{T\sqrt{1-2by}} \left( \left(\alpha + \frac{1}{T}\right)\sqrt{1-2by} - x - \frac{1}{T} \right) \\
&\hspace{52mm}+ 2\left( 1 - \sqrt{1-2y} \right) \left( x + \frac{1}{2T}\left( 1 - \sqrt{1-2by} \right)\right) (\alpha+1/T)\frac{1}{\sqrt{1-2by}} \Bigg] \\
&\hspace{-18mm}= \frac{ x\left((\alpha+1/T)\sqrt{1-2by} -x - 1/T\right) + (1-\sqrt{1-2by})\left( (\alpha+1/T) \left( \sqrt{1-2by} + 2x + (1-\sqrt{1-2by})/T \right) - x - 1/T \right) }{\sqrt{1-2by} \left( \left(\alpha + \frac{1}{T}\right)\sqrt{1-2by} - x - \frac{1}{T} \right)^3} \cdotp
\end{align*}

\section{Study of the Zhou-Kadanoff equilibrium}

\subsection{Establishing the symmetrized system in the $(x,w)$ variables}
\label{APPENSS___SymmeSystm(x,w)}

Starting from the equations \eqref{EQUATSS4.2ReducSystmSymme} page \pageref{EQUATSS4.2ReducSystmSymme}, and using the definition of the new intermediate variable $w$:
\begin{align*}
w = 1 -  \sqrt{1-2by},
\end{align*}
we find, on the one hand for the first variable:
\begin{align*}
x' &= \frac{r\sqrt{1-2by}}{(\alpha+1)\sqrt{1-2by}-x-1}\\
&= \frac{r(1-w)}{\alpha+1)(1-w) - x - 1}\\
&= \frac{r(w-1)}{(\alpha+1)w - \alpha - 1 + x + 1} \\
&= \frac{r(w-1)}{(\alpha+1)w + x - \alpha},
\end{align*}
which is equation \eqref{EQUATSS4.2SystmSymme_x_w1}, and on the other hand, for the second variable:
\begin{align*}
y' &= \frac{(1-\sqrt{1-2y})}{b} \cdot \frac{(x + \frac{1}{2}(1 - \sqrt{1-2by}))}{\left( (\alpha+1) \sqrt{1-2by} - x - 1 \right)^2} \\
&= \frac{w}{b} \cdotp \frac{(x + \frac{w}{2})}{\left( (\alpha+1) (1-w) - x - 1 \right)^2} \\
&= \frac{w}{b} \cdotp \frac{(x + \frac{w}{2})}{\left( -(\alpha+1)w - x + \alpha \right)^2} \\
&= \frac{w}{b} \cdotp \frac{(x + \frac{w}{2})}{\left( (\alpha+1)w + x - \alpha \right)^2},
\end{align*}
which is equation \eqref{EQUATSS4.2SystmSymme_x_w2}. The hybrid form \eqref{EQUATSS4.2SystmSymme_x_w_} of the symmetrized iteration is obtained.

\subsection{The infinite branch of the cubic $\mathcal{C}$ is a decreasing function}
\label{APPENSS___BrnchCubicDecrs}

In this section, let us study the global properties of the cubic $\mathcal{C}$ defined by the equation \eqref{EQUATSS4.2SystmFxdPtFina2} page \pageref{EQUATSS4.2SystmFxdPtFina2}, that is:
\begin{align*}
2x^2\left(x + \frac{w}{2}\right) = r^2(2-w)(w-1)^2,
\end{align*}
in the quadrant $\mathbb{R}_+\times]-\infty,1]$, here in the $(x,w)$-plane. More precisely, let us show that there exists a single infinite branch of this cubic in the quadrant, that is the graph of a function in $x$ which is strictly decreasing.\\
In particular, let us note that for all $x \geq 0$, there exists at least one solution $w$ to \eqref{EQUATSS4.2SystmFxdPtFina2}: the cubic $\mathcal{C}$ has an infinite branch which is unbounded in the direction $x \rightarrow +\infty$. Let us note that the cubic $\mathcal{C}$ intersects the line $w=1$ only at the point $(1,0)$ in the quadrant $\mathbb{R}_+\times]-\infty,1]$, and in the same way considering the intersection between the cubic $\mathcal{C}$ and the line $x=0$ leads to the equation $0=r^2(2-w)(w-1)^2$, so that the only intersection between this line and the cubic $\mathcal{C}$ in the quadrant is $(0,1)$. Now, the cubic $\mathcal{C}$ is crunodal (that is, the curve has a self-intersection), and two branches of this curve pass through the point $(0,1)$, with respective tangent vectors $(1,\pm1/r)$. Therefore, the branch entering the quadrant has a slope $-1/r$ at $x=0$, and such a branch has to remain cannot escape the quadrant. We conclude that the infinite branch of the cubic $\mathcal{C}$ is indeed contained in the quadrant $\mathbb{R}_+\times]-\infty,1]$.\\
It remains to study the variations of this branch in the quadrant. Writing the equation of the cubic $\mathcal{C}$ as $g(x,w) = 0$ where $g:(x,w) \mapsto 2x^2\left(x + \frac{w}{2}\right) - r^2(2-w)(w-1)^2$, and computing the gradient of $g$ we find:
\begin{align*}
\nabla g(x,w) = \begin{pmatrix}
6x^2 + 2xw
\\
r^2(w-1)^2 + 2r^2(w-1)(w-2)+x^2
\end{pmatrix}
=
\begin{pmatrix}
2x(3x+w) \\
r^2(w-1)(3w-5) + x^2
\end{pmatrix}.
\end{align*}
On the one hand the quadrant is always below the line $y= 5/3$, so that $\partial_w g(x,w) > 0$ for all points $(x,w)$ in the interior of the quadrant, and on the other hand since the equation:
\begin{align*}
2 x^2 \left(x - \frac{3x}{2}\right) = r^2 (2+3x)(-3x-1)^2,
\end{align*}
that is
\begin{align*}
- x^3 = r^2 (2+3x)(3x+1)^2
\end{align*}
has no real solution, the cubic $\mathcal{C}$ never intersects the line $w = -3x$, and lies actually above this line, we have $\partial_x g(x,w) > 0$ for all the points $(x,w)$ of the cubic $\mathcal{C}$ with $x>0$. The gradient of $g$ lying always in the first quadrant $x > 0,w > 0$ when $(x,w)$ is in the interior of the quadrant $\mathbb{R}_+\times]-\infty,1]$, we deduce that the infinite branch in this quadrant is actually the graph of a function $f_2:x\mapsto g_2(x)$, which is strictly decreasing.\\
The assertion concerning the behaviour of the cubic $\mathcal{C}$ in the quadrant $\mathbb{R}_+\times]-\infty,1]$ is proved.

\subsection{Location of the inflection point of the cubic $\mathcal{C}$}
\label{APPENSS___InflcPointCubic}

In the present section we will determine the position of the single inflection point of the cubic $\mathcal{C}$, defined by the equation \eqref{EQUATSS4.2SystmFxdPtFina2} page \pageref{EQUATSS4.2SystmFxdPtFina2}. As in the previous section, writing the equation of the cubic $\mathcal{C}$ as $g(x,w) = 0$ where $g:(x,w) \mapsto 2x^2\left(x + \frac{w}{2}\right) - r^2(2-w)(w-1)^2$, let us first obtain an equation for the locus of the inflection points of the cubic. Since the cubic $\mathcal{C}$ is described as an implicit equation, we can use the formula giving the inflection points for general implicit curves (such a formula can be obtained, for instance, thanks to the Lagrange's multipliers theorem). This general result states that $(x,w)$ is an inflection point of the curve $g(x,w)=0$ if and only if:
\begin{align}
\label{EQUATAPD.3PointInfleGener}
\left[2\partial_x\partial_wg \partial_x g \partial_w g - \partial_x^2g\left(\partial_wg\right)^2 - \partial_w^2g\left(\partial_xg\right)^2\right](x,w) = 0.
\end{align}
Let us call the curve defined by \eqref{EQUATAPD.3PointInfleGener} the \emph{curve of the critical points of the curve $g(x,w)=0$}.\\
In the present case, we have:
\begin{align*}
\partial_xg(x,w) = 2x(3x+w) \hspace{3mm} \text{and} \hspace{3mm} \partial_wg(x,w) = x^2 + r^2(w-1)(3w-5),
\end{align*}
so that
\begin{align*}
\partial_x^2 g(x,w) &= 2(6x+w),\\
\partial_x\partial_w g(x,w) &= 2x,\\
\partial_w^2 g(x,w) &= 2r^2(3w-4).
\end{align*}
Therefore the equation \eqref{EQUATAPD.3PointInfleGener} in our case writes:
\begin{align}
\label{EQUATAPD.3PointInfleCubiq}
8x^2(3x+w)\left( x^2 + r^2(w-1)(3w-5) \right) - 2(6x+w)\left( x^2 + r^2(w-1)(3w-5) \right)^2 - 8r^2x^2(3w-4)(3x+w)^2 = 0.
\end{align}
Our result concerning the location of the inflection point of the cubic $\mathcal{C}$ relies on the following observation: the cubic $\mathcal{C}$ intersects its own curve of critical points on the line $w = 3(x+1)$. This fact, which seems to be a miracle at first glance, is most probably a consequence of the fact that the family of cubics $\left(\mathcal{C}\right)_r$ (seen as a family of curves depending on the parameter $r$) is obtained by a family of affine dilatations of a given oblique strophoid.\\
The intersection points $(x,w)$ between the line $w=3(x+1)$ and the curve of the critical points of $\mathcal{C}$ are such that:
\begin{align}
\label{EQUATAPD.3InterCsingDroit}
8x^2(6x+3)\left(x^2+r^2(3x+2)(9x+4)\right) - &\left( x^2 + r^2(3x+2)(9x+4) \right)^2 2(9x+3)\nonumber\\
&\hspace{25mm}-4x^2(6x+3)^2 2r^2(9x+5) = 0.
\end{align}
Let us denote by $P_r$ the polynomial such that the previous equation writes $P_r(x) = 0$. In the same way, considering the intersections $(x,w)$ between the line $w=3(x+1)$ and the cubic curve $\mathcal{C}$, we find the equation $Q_r(x) = 0$, where $Q_r$ is the polynomial defined as:
\begin{align}
\label{EQUATAPD.3InterCubiqDroit}
x^2(5x+3) - r^2(-3x-1)(3x+2)^2.
\end{align}
As a matter of fact, $Q_r$ is a factor of $P_r$, and we have:
\begin{align*}
P_r(x) = Q_r(x) \left( (6-486r^2)x^2 - 436r^2x - 96r^2 \right).
\end{align*}
Therefore, any intersection point $(x,w)$ between the cubic $\mathcal{C}$ and the line $w = 3(x+1)$ is also a point of intersection between the curve of the critical points of $\mathcal{C}$, and this line. By uniqueness of the inflection point of the cubic $\mathcal{C}$, we deduce that the inflection point of $\mathcal{C}$ is on the line $w=3(x+1)$.\\
It remains then to study the equation $Q_r(x)=0$ to determine the abscissa of the inflection point of the cubic. $Q_r$ can be rewritten as:
\begin{align*}
Q_r(x) = (27r^2 + 5)x^3 + (45r^2+3)x^2 + 24r^2x+4r^2.
\end{align*}
The derivative of the cubic polynomial $Q_r$ is
\begin{align*}
Q_r'(x) = 3(27r^2+5)x^2 + 2(45r^2+3)x+24r^2,
\end{align*}
and the discriminant of the quadratic polynomial $Q_r'$ is:
\begin{align*}
\Delta_{Q'_r} = 36\left(9r^4-10r^2+1\right) = 36(r^2-1)(r^2-1/9).
\end{align*}
For $r \geq 1/3$, $Q_r'$ has at most one real single root, and so the cubic $Q_r$ is strictly increasing, and has a single root. Since $Q_r(0) = 4r^2 > 0$, this root is negative, and so $Q_r$ is positive on $\mathbb{R}_+$.\\
For $r < 1/3$, $Q_r'$ has two real roots, the largest being:
\begin{align*}
x_{Q'}^+ = \frac{-(15r^2+1) + \sqrt{9r^4-10r^2+1}}{27r^2+5}\cdotp
\end{align*}
The numerator of $x_{Q'}^+$ is a strictly negative quantity for $0 < r < 1$, so that $Q_r$ is again strictly increasing on $[0,+\infty[$. Since $Q_r(0) > 0$, we deduce that, in this case as well, $Q_r(x)$ never vanishes for $x \geq 0$.\\
As a consequence, the inflection point of the cubic $\mathcal{C}$ lies in the half-plane $\mathbb{R}_-^*\times\mathbb{R}$.\\
\end{appendices}
\newline
\noindent
\textbf{Acknowledgements.} The authors are grateful to E. Caglioti, I. Gallagher, B. Lods, M. Pulvirenti, C. Saffirio and S. Simonella for many stimulating discussions concerning the topic of the present article. The authors gratefully acknowledge the financial support of the Hausdorff Research Institute for Mathematics (Bonn) through the collaborative research center The mathematics of emerging effects (CRC 1060, Project-ID 211504053), and the Deutsche Forschungsgemeinschaft (DFG, German Research Foundation).

E-mail address: \texttt{dolmaire@iam.uni-bonn.de}, \texttt{velazquez@iam.uni-bonn.de}.


\begin{thebibliography}{99}

\bibitem{Alex975} Roger K. Alexander, \emph{The Infinite Hard-Sphere System}, Ph.D thesis, University of California in Berkeley (1975).


\bibitem{BeCa999} Dario Benedetto, Emanuele Caglioti, ``The collapse phenomenon in one-dimensional inelastic point particle systems'', \emph{Physica D}, \textbf{132}, 457--475 (1999).

\bibitem{BeMa990} Bernard Bernu, Redha Mazighi, ``One-dimensional bounce of inelastically colliding marbles'', \emph{Journal of Physics A: Mathematical and General}, \textbf{23}, 5745--5754 (1990).

\bibitem{BrPo004} Nikolai V. Brilliantov, Thorsten P\"oschel, \emph{Kinetic Theory of Granular Gases}, Oxford University Press (2004).

\bibitem{CHMR021} Jos\'e A. Carrillo, Jingwei Hu, Zheng Ma, Thomas Rey, ``Recent Development in Kinetic Theory of Granular Materials: Analysis and Numerical Methods'', in \emph{Trails in Kinetic Theory}, SEMA SIMAI Springer Series, \textbf{25}, 1--36, Springer-Verlag (2021).

\bibitem{ChKZ022} Bernard Chazelle, Kritkorn Karntikoon, Yufei Zheng, ``A geometric approach to inelastic collapse'', \emph{Journal of Computational Geometry}, \textbf{13}:1, 197--203 (2022).

\bibitem{CDKK999} Barry A. Cipra, Paolo Dini, Stephen Kennedy, Amy Kolan, ``Stability of one-dimensional inelastic collision sequences of four balls'', \emph{Physica D}, \textbf{125}, 183--200 (1999).

\bibitem{CoGM995} Peter Constantin, Elizabeth Grossman, Muhittin Mungan, ``Inelastic collisions of three particles on a line as a two-dimensional billiard'', \emph{Physica D}, \textbf{83}, 409--420 (1995).

\bibitem{DoVeAr1} Th\'eophile Dolmaire, Juan J. L. Vel\'azquez, ``About the collapse of three inelastic particles in dimension $d \geq 2$'', preprint arXiv:2402.13803v2 (02/2024).

\bibitem{DoVeNot} Th\'eophile Dolmaire, Juan J. L. Vel\'azquez, ``A particle model that conserves the measure in the phase space, but does not conserve the kinetic energy'', preprint arXiv:2403.02162 (03/2024).

\bibitem{Draz005} Philip G. Drazin, \emph{Nonlinear systems}, Cambridge texts in applied mathematics, \textbf{10}, Cambridge University Press (2005).

\bibitem{GSRT013} Isabelle Gallagher, Laure Saint-Raymond, Benjamin Texier, \emph{From Newton to Boltzmann: Hard Spheres and Short-Range Potentials}, Zurich Lectures in Advanced Mathematics, \textbf{18}, European Mathematical Society (EMS), Z\"urich (2013).

\bibitem{Gibs998} Christopher G. Gibson, \emph{Elementary Geometry of Algebraic Curves: An Undergraduate Introduction}, Cambridge University Press (1998).

\bibitem{GoZa993} Isaac Goldhirsch, Gianluigi Zanetti, ``Clustering Instability in Dissipative Gases'', \emph{Physical Review Letters}, \textbf{70}:11, 1619-1622 (03/1993).

\bibitem{GSBM998} Daniel I. Goldman, Mark D. Shattuck, Chris Bizon, William D. McCormick, Jack B. Swift, Harry L. Swinney, ``Absence of inelastic collapse in a realistic three ball model'', \emph{Physical Review E}, \textbf{57}:4, 4831--4833 (04/1998).

\bibitem{GrMu996} Elizabeth Grossman, Muhittin Mungan, ``Motion of three inelastic particles on a ring'', \emph{Physical Review E}, \textbf{53}:6, 6435--6449 (06/1996).

\bibitem{HuRo023} Eleni H\"ubner-Rosenau, ``Some Problems in Particle Systems: Inelastic Hard Spheres'', Master thesis, Mathematisch-Naturwissenschaftliche Fakultät der Rheinischen Friedrich-Wilhelms-Universit\"at Bonn (2023).

\bibitem{BeJN996} Heinrich M. Jaeger, Sidney R. Nagel, Robert P. Behringer, ``Granular solids, liquids, and gases'', \emph{Reviews of Modern Physics}, \textbf{68}:4, 1259--1273 (10/1996).


\bibitem{McYo991} Sean McNamara, William R. Young, ``Inelastic collapse and clumping in a one-dimensional granular medium'', \emph{Physics of Fluids A: Fluids Dynamics}, \textbf{4}:3, 496--504 (03/1992).


\bibitem{McYo993} Sean McNamara, William R. Young, ``Inelastic collapse in two dimensions'', \emph{Physical Review E}, \textbf{50}:1, R28--31 (07/1994).


\bibitem{PoSc005} Thorsten P\"oschel, Thomas Schwager, \emph{Computational Granular Dynamics: Models and Algorithms}, Springer-Verlag (2005).

\bibitem{ScZh996} Norbert Sch\"orghofer, Tong Zhou, ``Inelastic collapse of rotating spheres'', \emph{Physical Review E}, \textbf{54}:5, 5511-5515 (11/1996).

\bibitem{ShKa989} Koichiro Shida, Toshio Kawai, ``Cluster formation by inelastically colliding particles in one-dimensional space'', \emph{Physica A}, \textbf{162}, 145--160 (1989).


\bibitem{ZhKa996} Tong Zhou, Leo P. Kadanoff, ``Inelastic collapse of three particles'', \emph{Physical Review E}, \textbf{54}:1, 623--628 (07/1996).

\end{thebibliography}
\end{document}